\documentclass[12pt]{article}

\usepackage[margin=1in]{geometry}

\usepackage{amsmath}
\usepackage{amssymb}
\usepackage{amsthm}
\usepackage{xspace}
\usepackage[normalem]{ulem}
\usepackage{graphicx}
\usepackage{epsfig}
\usepackage{ifpdf}
\usepackage{url,hyperref}
\usepackage{latexsym}
\usepackage[mathscr]{euscript}
\usepackage{xspace}
\usepackage{color}
\usepackage{makeidx}
\usepackage{wrapfig,color}
\usepackage{stackrel}
\usepackage{algorithm}

\usepackage[all]{xy}
\usepackage{framed}
\usepackage{pb-diagram}

\long\def\remove#1{}

\usepackage{algpseudocode}

\newtheorem{theorem}{Theorem}[section] 
\newtheorem{lemma}[theorem]{Lemma}
\newtheorem{claim}{Claim}
\newtheorem{corollary}[theorem]{Corollary}
\newtheorem{proposition}[theorem]{Proposition}
\newtheorem{definition}{Definition}

\newtheorem{remark}{Remark}[section]
\newtheorem{condition}{Condition}[section]

\newcommand\id{\mathrm{id}}
\newcommand\cc{\mathrm{cc}}
\newcommand\mC{\mathrm{C}}
\newcommand\trunc{\mathrm{Tr}}

\newcommand\N{\mathbb{N}}
\newcommand\R{\mathbb{R}}
\newcommand\Z{\mathbb{Z}}
\newcommand\powset{\mathrm{pow}}
\newcommand\diam{\mathrm{diam}}

\newcommand\MM{\mathrm{MM}}
\newcommand\M{\mathrm{M}}

\newcommand\cech	{\mathrm{{C}ech}} 

\newcommand\res{\mathrm{res}}
\newcommand\vrt{\mathrm{V}}
\newcommand\edg{\mathrm{E}}

\newcommand{\initsampling}		{{\nu}}

\newcommand{\mycechconst}	{{c(s+2)}}

\definecolor{darkblue}{rgb}{0.0, 0.0, 0.8}
\definecolor{darkred}{rgb}{0.8, 0.0, 0.0}
\definecolor{darkgreen}{rgb}{0.0, 0.8, 0.0}

\newcommand{\eps}               {{\varepsilon}}
\newcommand{\denselist}{\itemsep 0pt\parsep=1pt\partopsep 0pt}

\title{Multiscale Mapper: Topological Summarization via Codomain Covers}

\author{Tamal K. Dey\thanks{Department of Computer Science and Engineering, The Ohio State University. \texttt{tamaldey, yusu@cse.ohio-state.edu}}, Facundo M\'emoli\thanks{Department of Mathematics and Department of Computer Science and Engineering, The Ohio State University. \texttt{memoli@math.osu.edu}}, Yusu Wang$^*$}

\date{}

\begin{document}

\maketitle

\begin{abstract}  
Summarizing topological information from datasets and maps defined on them is a central theme in topological data analysis. \textsf{Mapper}, a tool for such summarization, takes as input both a possibly high dimensional dataset and a map defined on the data, and produces a summary of the data by using a cover of the codomain of the map. This cover, via a pullback operation to the domain, produces 
a simplicial complex connecting the data points. 

{The resulting view of the data through a cover of the codomain offers flexibility in analyzing the data. However, it offers only a view at a fixed scale at which the cover is constructed.}
Inspired by the concept, we explore a notion of a 
tower of covers which induces a 
tower of simplicial complexes connected by simplicial maps, which we call {\em multiscale mapper}. 
We study the resulting structure, its stability,
and design practical algorithms to compute its persistence diagrams efficiently. 
Specifically, when the domain is a simplicial complex and the map is a real-valued piecewise-linear function, the algorithm can compute the exact persistence diagram only from the 1-skeleton of the input complex. 
For general maps, we present a combinatorial version of the algorithm that acts only on \emph{vertex sets} connected by the 1-skeleton graph, and this algorithm approximates the exact persistence diagram thanks to a stability result that we show to hold. We also relate the multiscale mapper with the \v{C}ech complexes arising from a natural pullback pseudometric defined on the input domain.

\end{abstract}

\setcounter{page}{1}

\tableofcontents

\section{Introduction}
\label{sec:intro}

Recent years have witnessed significant progress in applying topological ideas to analyzing complex and diverse data.  Topological ideas can be particularly powerful in deriving a succinct and meaningful summary of input data.  For example, the theory of \emph{persistent homology} built upon
\cite{ELZ02,frosini,vanessa}  and other fundamental developments 
\cite{circlemap,multi-dim,CCGGO09,CS10,CEH07,extended,DFW-sph,ZC05},  has provided a powerful and flexible framework for summarizing information of an input space or a scalar field into a much simpler object called the persistence diagram/barcode.

Modern data can be complex both in terms of the domain where they come from and in terms of properties/observations associated with them which are
often modeled as maps. 
For example, we can have a set of patients, where each patient is associated with multiple biological markers, giving rise to a multivariate map from the space of patients to an image domain that may or may not be the Euclidean space. 
To this end, we need to develop theoretically justified methods to analyze not only real-valued scalar fields, but also more complex maps defined on a given domain, such as multivariate, circle valued, sphere valued maps, etc.  

There has been interesting work in this direction, including multidimensional persistence \cite{multi-dim,betti-landi} and persistent homology for circular valued maps \cite{circlemap}. However, summarizing multivariate maps using these techniques appears to be challenging. 
Our approach takes a different direction, and is inspired by and based on the mapper methodology, recently introduced by Singh et al.~ in \cite{mapper}. 
Taking an observation made in \cite{mapper} regarding 
the behavior of Mapper under a change in the covers as a starting point, 
we study a multiscale version of mapper, which we will henceforth refer to as 
{\em multiscale mapper}, that is capable of producing a multiscale 
summary using a cover of the codomain at different scales.

Given a map $f: X \to Z$, Singh et al. proposed a novel concept to create a topological metaphor, called the \emph{mapper}, for the structure behind $f$ by pulling back a cover of the space $Z$ to a cover on $X$ through $f$. 
This mapper methodology is general: it can work with any (reasonably tame) continuous maps between two topological spaces, and it converts complex maps and covers of the target space, into simplicial complexes, which are much easier to manipulate computationally. It is also powerful and flexible-- one can view the map $f$ and a finite cover of the space $Z$ as the lens through which the input data $X$ is examined. By choosing different maps and covers, the resulting mapper representation captures different aspects of the input data. Indeed, the mapper methodology has been successfully applied to analyzing various types of data, see e.g. \cite{survey,cancer}, and it is a main component behind the data analytics platform developed by the company Ayasdi.

\subsection{Contributions}
Given an input map $f: X \to Z$ and a finite cover $\mathcal{U}$ of $Z$, the induced mapper $\mathrm{M}(\mathcal{U},f)$ is a simplicial complex encoding the structure of $f$ through the lens of $Z$. 
However, the simplicial complex $\mathrm{M}(\mathcal{U},f)$ provides only \emph{one} snapshot of $X$ at a \emph{fixed} scale as determined by the scale of the cover $\mathcal{U}$. 
Using the idea of persistence homology, we study the evolution of the mapper $\mathrm{M}(f, \mathcal{U}_\eps)$ for a \emph{tower of covers} $\mathfrak{U} = \{ \mathcal{U}_\eps \}_\eps$ at multiple scales (indexed by $\eps$). 

As an intuitive example, consider a real-valued function $f: X \to \mathbb{R}$, and a cover $\mathcal{U}_\eps$ of $\mathbb{R}$ consisting of all possible intervals of length $\eps$. Intuitively, as $\eps$ tends to 0, the corresponding Mapper $\mathrm{M}(f, \mathcal{U}_\eps)$ approaches the Reeb graph of $f$. As $\eps$ increases, we look at the Reeb graph at coarser and coarser resolution. The multiscale mapper in this case roughly encodes this simplification process.

The \emph{multiscale mapper} $\MM(\mathfrak{U},f)$, which we formally define in \S \ref{sec:multimapper}, consists of a sequence of simplicial complexes connected with simplicial maps. Upon passing to homology with fields coefficients, the information in $\MM(\mathfrak{U},f)$ can be further summarized in the corresponding persistence diagram. 
In other words, we can now summarize an input described
by a multivariate (or circle/sphere valued) map into a single persistence diagram, much like in traditional persistence homology for real-valued functions. 

In  \S\ref{sec:stability}, we discuss the stability of the multiscale mapper, under changes in the input map and/or in the tower of covers $\mathfrak{U}$. Stability is a  highly desirable property for a summary as it implies robustness to noise in data and in measurements. Interestingly, analogous to the case of homology versus persistence homology, mapper does not satisfy a stability property, whereas multiscale mapper does enjoy stability as we show in this paper.

To facilitate the broader usage of mapper and multiscale mapper as a data analysis tool, we develop efficient algorithms for computing
and approximating mapper and multiscale mapper. 
In particular, even for piecewise-linear functions defined on a simplicial complex, 
the standard algorithm 
needs to determine for each simplex the subset (partial simplex)
on which the function value falls within a certain range.

In \S \ref{sec:exactcomp}, we show that for such an input, it
is sufficient to consider the restriction of the function to the 1-skeleton of 
the complex for computing the mapper and the multiscale mapper. Since the
$1$-skeleton (a graph) is typically much smaller in size than the full complex, this helps improving the time efficiency of  
computing the mapper and multiscale mapper outputs. 

In \S \ref{sec:approxcomp}, 
we consider the more general case of a map $f: X \to Z$ where $X$ is a simplicial complex but $Z$ is not necessarily real-valued. We show that there is an even simpler combinatorial version of the multiscale mapper, which only acts on \emph{vertex sets} of $X$ with connectivity given by the 1-skeleton graph of $X$ \footnote{We note that a variant of this combinatorial version is what is currently used in the publicly available software for mapper in practice.}. 
The cost we pay here is that the resulting persistence diagram \emph{approximates} (instead of computing exactly) that of the standard multiscale mapper, and the tower of covers of $Z$ needs to satisfy a ``goodness" condition.

In \S \ref{sec:metricview}, we show that given a tower of covers $\mathfrak{U}$ and a map $f:X\rightarrow Z$ there exists a natural pull-back pseudo-metric $d_{\mathfrak{U},f}$ defined on the input domain $X$. With such a pseudo-metric on $X$,  we can now construct the standard \v{C}ech filtration $\mathfrak{C}=\{\cech_\eps(X)\}_\eps$ (or Rips filtration) in $X$ directly, instead of computing the Nerve complex of the pull-back covers as required by mapper.
The resulting filtration $\mathfrak{C}$ is \emph{connected by inclusion maps} instead of {\em simplicial maps}.  This is easier for computational purposes even though one
has a method to compute the persistence
diagram of a filtration involving arbitrary simplicial maps \cite{DFW-sph}.
Furthermore, it turns out that the resulting sequence of \v{C}ech complexes $\mathfrak{C}$ interleaves with the sequence of complexes $\MM(\mathfrak{U},f)$, implying that their corresponding persistence diagrams approximate each other.

Some technical details and proofs are relegated to the Appendix. 

\section{Topological background and motivation}\label{sec:background}

In this section we recall several facts about topological spaces and simplicial complexes \cite{munkres}. Let $K$ and $L$ be two finite simplicial complexes over the vertex sets $V_K$ and $V_L$, respectively. A set map $\phi:V_K\rightarrow V_L$ is a \emph{simplicial map} if $\phi(\sigma)\in L$ for all $\sigma\in K$. 



By an open cover of a topological space $X$ we mean a collection $\mathcal{U}=\{U_\alpha\}_{\alpha\in A}$ of open sets such that $\bigcup_{\alpha\in A} U_\alpha= X.$ In this paper, whenever referring to an open cover, we will always assume that each $U_\alpha$ is path connected.

\begin{definition}[Nerve of a cover]
Given a finite cover ${\mathcal U} = \{
 U_{\alpha}\}_{\alpha \in A}$ of a topological space $X$, we define the {\em
 nerve} of the cover ${\mathcal U}$ to be the simplicial complex
 $N({\mathcal U})$ whose vertex set is the index set $A$, and where
 a subset $\{ \alpha _0 , \alpha _1, \ldots , \alpha _k \}\subseteq A$ spans a
 $k$-simplex in $N({\mathcal U})$ if and only if $U_{\alpha _0 } \cap
 U_{\alpha _1 } \cap \ldots \cap U_{\alpha _k} \neq \emptyset$.  
\end{definition}

 \noindent Suppose that we are given a topological space $X$ equipped with a
 continuous map $f: X \rightarrow Z$ into a parameter space $Z$, 
where $Z$ is equipped with an open cover ${\mathcal U} = \{
 U_{\alpha} \}_{\alpha \in A}$
for some finite index set
 $A$.  Since $f$ is continuous, the sets $\{f^{-1}(U_{\alpha}),\,\alpha\in A\}$ 
 form an open cover of $X$.  For each $\alpha$, we can now consider
 the decomposition of $f^{-1}(U_{\alpha })$ into its path connected
 components, so we write $f^{-1}(U_{\alpha }) = \bigcup _{i =
 1}^{j_{\alpha}} V_{\alpha , i }$, where $j_{\alpha}$ is the number of
 path connected components $V_{\alpha,i}$'s in $f^{-1}(U_{\alpha })$.  We write
 $f^\ast({\mathcal U}) $ for the cover of $X$ obtained this way
 from the cover ${\mathcal U}$ of $Z$ and refer to it as the \emph{pullback} cover of $X$ induced by $\mathcal{U}$ via $f$.

Notice that there are pathological examples of $f$
where $f^{-1}(U_\alpha)$ may shatter into infinitely many path components. This motivates us to
consider \emph{well-behaved} functions $f$: we require that for every 
path connected open set $U\subseteq Z$, the preimage $f^{-1}(U)$ has \emph{finitely} many path connected components.
An example of well-behaved functions is given by piecewise-linear real valued functions defined on finite simplicial complexes. It follows that for any well-behaved function $f:X\rightarrow Z$ and any finite open cover $\mathcal{U}$ of $Z$, the open cover $f^\ast(\mathcal{U})$ is also finite. 

\begin{quote}
\emph{If not stated otherwise, all functions and all covers are assumed to be well-behaved and finite, respectively. Consequently, all nerves of pullbacks of finite covers will also be finite.} 
\end{quote}


\begin{definition}[Mapper \cite{mapper}]\label{def:mapper}
Let $X$ and $Z$ be topological spaces and let $f:X\rightarrow Z$ be a well-behaved and continuous map. Let $\mathcal{U} = \{U_\alpha\}_{\alpha\in A}$ be a finite open cover of $Z$.
The \emph{mapper construction} arising from these data is defined to be the nerve simplicial complex of the pullback cover:
$\mathrm{M}(\mathcal{U},f) := {N}(f^\ast(\mathcal{U})).$ 
\end{definition}


\begin{remark}[For Definition \ref{def:mapper}]
This construction is quite general. It encompasses both the Reeb graph and merge trees at once: consider $X$ a topological space and $f:X\rightarrow \mathbb{R}$. Then, consider the following two options for $\mathcal{U}=\{U_\alpha\}_{\alpha\in A}$, the other ingredient of the construction:
\begin{itemize}
\item $U_\alpha = (-\infty,\alpha)$ for $\alpha\in A = \mathbb{R}$. This corresponds to sublevel sets which in turn lead to merge trees.
\item $U_\alpha = (\alpha-\varepsilon,\alpha+\varepsilon)$ for $\alpha\in A=\mathbb{R}$, for some fixed $\varepsilon>0$. This corresponds to ($\varepsilon$-thick) level sets, which induce a relaxed notion of Reeb graphs.
\end{itemize}
In these two examples, for simplicity of presentation, the set $A$  is allowed to have infinite cardinality. Also, note one can take \emph{any} open cover of $\mathbb{R}$ in this definition. This may give rise to other constructions beyond merge trees or Reeb graphs. For instance, one may choose any point $r\in \mathbb{R}$ and let $U_\alpha=(r-\alpha,r+\alpha)$ for each $\alpha\in A=\mathbb{R}$ or other constructions.
\end{remark}

\subsection*{Maps between covers.}
 If we have two covers ${\mathcal U} = \{ U_{\alpha} \} _{\alpha \in A}$ and ${\mathcal V} = \{ V_{\beta} \} _{\beta \in B}$ of a space $X$, a {\em map of covers} from ${\mathcal U} $ to ${\mathcal V}$ is a set map $\xi: A \rightarrow B$ so that $U_{\alpha} \subseteq V_{\xi(\alpha)}$ for all $\alpha \in A$. By an abuse of notation we also use $\xi$ to indicate the map $\mathcal{U}\rightarrow\mathcal{V}.$  Given such a map of covers, there is an induced simplicial map $N(\xi) : N({\mathcal U}) \rightarrow N({\mathcal V})$, given on vertices by the map $\xi$.  
Furthermore, if $\mathcal{U}\stackrel{\xi}{\rightarrow}\mathcal{V}\stackrel{\zeta}{\rightarrow}\mathcal{W}$ are three different covers of a topological space with the intervening maps of covers between them, then $N(\zeta\circ\xi) = N(\zeta)\circ N(\xi)$ as well.

%
%

The following simple lemma will be very useful later on.

\begin{lemma}[Induced maps are contiguous]\label{lemma:cover-contiguity}
Let $\zeta,\xi:\mathcal{U}\rightarrow\mathcal{V}$ be any two maps of covers. Then, the simplicial maps $N(\zeta)$ and $N(\xi)$ are contiguous.
\end{lemma}

Recall that two simplicial maps $h_1, h_2: K \rightarrow L$ are \emph{contiguous} if for all $\sigma\in K$ it holds that $h_1(\sigma)\cup h_2(\sigma)\in L$. In particular, contiguous maps induce identical maps at the homology level \cite{munkres}. Lemma \ref{lemma:cover-contiguity} implies that the map $H_\ast(N(\mathcal{U}))\rightarrow H_\ast(N(\mathcal{V}))$ thus induced can be deemed canonical.

\begin{proof}[Proof of lemma \ref{lemma:cover-contiguity}]
Write $\mathcal{U}=\{U_\alpha\}_{\alpha\in A}$ and $\mathcal{V}=\{V_\beta\}_{\beta\in B}.$ Then, for all $\alpha\in A$ we have both
$$U_\alpha\subseteq V_{\zeta(\alpha)}\hspace{.1in}\mbox{and}\hspace{.1in}U_\alpha\subseteq V_{\xi(\alpha)}.$$
This means that $U_\alpha\subseteq V_{\zeta(\alpha)}\cap V_{\xi(\alpha)}$ for all $\alpha\in A$. Now take any $\sigma\in N(\mathcal{U})$. We need to prove that $\zeta(\sigma)\cup\xi(\sigma)\in N(\mathcal{V}).$ For this write

\begin{align*}\bigcap_{\beta\in\zeta(\sigma)\cup\xi(\sigma)} V_\beta &= \left(\bigcap_{\alpha\in \sigma} V_{\zeta(\alpha)}\right) \cap \left(\bigcap_{\alpha\in \sigma} V_{\xi(\alpha)}\right)\\
&=\bigcap_{\alpha\in \sigma} \left(V_{\zeta(\alpha)}\cap V_{\xi(\alpha)} \right)\\
&\supseteq \bigcap_{\alpha\in \sigma} U_\alpha \neq \emptyset,\end{align*}

where the last step follows from assuming that $\sigma\in N(\mathcal{U}).$
\end{proof}

\subsection*{Pullbacks.} 
When we consider a space $X$ equipped with a continuous map $f: X \rightarrow Z$ to
 a topological space $Z$, and we are given a map of covers $\xi:{\mathcal
 U} \rightarrow {\mathcal V}$ between covers of $Z$, there is a corresponding map of
 covers between the respective pullback covers of $X$: 
$f^\ast(\xi):f^\ast({\mathcal U}) \longrightarrow f^\ast({\mathcal
 V}).$ 

Indeed, we only need to note that if $U
 \subseteq V$, then $f^{-1}(U) \subseteq f^{-1}(V)$, and
 therefore it is clear that each path connected component of $f^{-1}(U)$
 is included in exactly one path connected component of $f^{-1}(V)$.  More precisely, let $\mathcal{U}=\{U_\alpha\}_{\alpha\in A}$, $\mathcal{V}=\{V_\beta\}_{\beta\in B}$, with $U_\alpha\subseteq V_{\xi(\alpha)}$ for $\alpha\in A$. 
Let $\widehat{U}_{\alpha,i}$, $i\in\{1,\ldots,n_\alpha\}$ denote  the connected components of $f^{-1}(U_\alpha)$ and $\widehat{V}_{\beta,j}$, $j\in\{1,\ldots,m_\beta\}$ denote the connected components of $f^{-1}(V_\beta)$.
Then, the map of covers $f^\ast(\xi)$ from $f^\ast({\mathcal U})$ to
 $f^\ast({\mathcal V})$ is given by requiring that each set
 $\widehat{U}_{\alpha,i}$ is sent to the \emph{unique} set of the form
 $\widehat{V}_{\xi(\alpha),j}$ so that $\widehat{U}_{\alpha,i} \subseteq \widehat{V}_{\xi(\alpha),j}$. 

Furthermore, if $\mathcal{U}\stackrel{\xi}{\rightarrow}\mathcal{V}\stackrel{\zeta}{\rightarrow}\mathcal{W}$ are three different covers of a topological space with the intervening maps of covers between them, then $f^\ast(\zeta\circ\xi) = f^\ast(\zeta)\circ f^\ast(\xi).$

\section{Multiscale Mapper}
\label{sec:multimapper}

In the definition below, \emph{objects} can be covers, simplicial complexes, or vector spaces.
\begin{definition}[Tower]\label{def:hfc}
A \emph{tower} $\mathfrak{W}$ with resolution $r\in\R$ 
is any collection 
$\mathfrak{W}=\big\{\mathcal{W}_\eps\big\}_{\eps\geq r}$
of objects $\mathcal{W}_\varepsilon$ together with maps $w_{\varepsilon,\varepsilon'}: {\mathcal W} _\varepsilon \rightarrow {\mathcal W}_{\varepsilon'}$  
so that $w_{\varepsilon,\varepsilon}=\id$ and 
$w_{\varepsilon',\varepsilon''}\circ w_{\varepsilon,\varepsilon'} = w_{\varepsilon,\varepsilon''}$ for all $r\leq \varepsilon\leq \varepsilon'\leq \varepsilon''$.
Sometimes we write
$\mathfrak{W}=\big\{\mathcal{W}_\varepsilon\overset{\tiny{w_{\varepsilon,\varepsilon'}}}{\longrightarrow} \mathcal{W}_{\varepsilon'}\big\}_{r\leq\varepsilon\leq \varepsilon'}$ to denote the collection with the maps. 
Given such a tower $\mathfrak{W}$, $\res(\mathfrak{W})$ refers to its resolution. 

When $\mathfrak{W}$ is a collection of \emph{finite} covers equipped with maps of covers between them, we call it \emph{a tower of covers}. 
When $\mathfrak{W}$ is a collection of \emph{finite} simplicial complexes equipped with simplicial maps between them, we call it \emph{a tower of simplicial complexes}. 
\end{definition}

The notion of resolution, and the variable $\varepsilon$, intuitively specify the granularity of the covers and the simplicial complexes induced by them. 

The pullback properties described at the end of \S\ref{sec:background} make it possible to take the pullback of a given  tower of covers of a space via a given continuous function into another space, so that we obtain:

\begin{proposition}\label{coro:pullback-cover}
Let $\mathfrak{U}$ be a tower of covers of $Z$ and $f:X\rightarrow Z$ be a continuous function.
Then, $f^\ast(\mathfrak{U})$ is a tower of covers of $X$.
\end{proposition}

In general, given a tower of covers $\mathfrak{W}$ of a space $X$, the nerve 
of each cover  in $\mathfrak{W}$ together with 
simplicial maps induced by each map of $\mathfrak{W}$ provides a tower of simplicial complexes which we denote by $N(\mathfrak{W})$. 
\begin{definition}[Multiscale Mapper]
Let $X$ and $Z$ be topological spaces and $f:X\rightarrow Z$ be a continuous map. Let $\mathfrak{U}$ be a tower of covers of $Z$.
Then, the \emph{multiscale mapper} is defined to be the tower of simplicial complexes defined by the nerve of the pullback:
$$\mathrm{MM}(\mathfrak{U},f):=N(f^\ast(\mathfrak{U})).$$
\end{definition}

Consider for example a sequence $\res(\mathfrak{U})\leq \varepsilon_1<\varepsilon_2<\ldots<\varepsilon_n$ of  $n$ distinct real numbers. Then, the definition of multiscale mapper $\mathrm{MM}(\mathfrak{U},f)$ gives rise to the following:
\begin{equation}\label{small-sc}
N(f^\ast(\mathcal{U}_{\varepsilon_1}))\rightarrow N(f^\ast(\mathcal{U}_{\varepsilon_2}))\rightarrow\cdots\rightarrow N(f^\ast(\mathcal{U}_{\varepsilon_n}))
\end{equation}
which is a sequence of simplicial complexes connected by simplicial maps. 

Applying to them the homology functor $\mathrm{H}_k(\cdot)$, $k=0,1,2,\ldots$, with coefficients in a field, one obtains persistence modules \cite{EH09}: tower 
 of vector spaces connected by linear maps for which efficient 
persistence algorithm is known \cite{DFW-sph}:
\begin{equation}\label{small-sc-vec}
\mathrm{H}_k\big(N(f^\ast(\mathcal{U}_{\varepsilon_1}))\big)\rightarrow \cdots\rightarrow \mathrm{H}_k\big(N(f^\ast(\mathcal{U}_{\varepsilon_n}))\big).
\end{equation}

More importantly, the information contained in a persistence module can
be summarized by  simple descriptors: its associated persistence diagrams. As pointed
out in~\cite{CCGGO09}, a finiteness condition is required, but given our assumptions that the covers are finite and that the function $f$ is well-behaved we obtain that the homology groups of all nerves have finite dimensions.  Now one can summarize the persistence module with a finite persistent diagram for the sequence $\MM(\mathfrak{U},f)$, denoted by $\mathrm{D_k}\mathrm{MM}(\mathfrak{U},f)$ for each dimension $k \in \N$ (see \cite{EH09} for background on persistence diagrams).

\section{Stability} 
\label{sec:stability}
To be useful in practice, the multiscale mapper should be stable
against the perturbations in the maps and the covers.
we show that such a stability is enjoyed by the multiscale mapper
under some natural condition on the tower of covers.  
The notion of stability in the context of
topological data analysis has been recently studied by
many researchers, see e.g.
\cite{bauer,bubenik,CCGGO09,struct,CEH07}.
In particular, Cohen-Steiner et al. \cite{CEH07} showed that 
persistence diagrams are \emph{stable} in the {\em bottleneck distance} meaning that small perturbations to a map yield small variations in the computed persistence diagrams.  In \cite{CCGGO09}, stability is expressed directly at the (algebraic) level of persistence modules (\ref{small-sc-vec}) via a quantitative structural condition called \emph{interleaving} of pairs of persistence modules. 
Before we state our stability results, we identify compatible notions of interleaving for tower of covers and tower of simplicial complexes, as a way to measure the ``closeness" between two tower of covers (or two tower of simplicial complexes). 

\subsection{Interleaving of towers of covers and simplicial complexes}
\label{sec:stab-cov}

In this section we consider 
towers of covers and simplicial complexes indexed over $\mathbb{R}$. In practice, we often have towers indexed by a discrete set in $\mathbb{R}$. 
Any such tower can be extended to a 
tower of covers (or simplicial complexes) indexed over $\mathbb{R}$ 
by taking the covers for any intermediate open interval 
$(\eps,\eps')\subset \mathbb{R}$ to be same as that at $\eps\in \mathbb{R}$.

\begin{definition}[Interleaving of towers of covers] Let $\mathfrak{U} = \{\mathcal{U}_\varepsilon\}$ and $\mathfrak{V}=\{\mathcal{V}_\varepsilon\}$ be two towers of covers of a topological space $X$ such that $\res(\mathfrak{U})=\res(\mathfrak{V})=r$. Given $\eta\geq 0$,  we say that $\mathfrak{U}$ and $\mathfrak{V}$ are $\eta$-interleaved if one can find maps of covers $\zeta_\varepsilon:\mathcal{U}_\varepsilon \rightarrow \mathcal{V}_{\varepsilon+\eta}$ and $\xi_{\varepsilon'}:\mathcal{V}_{\varepsilon'} \rightarrow \mathcal{U}_{\varepsilon'+\eta}$ for all $\varepsilon,\varepsilon'\geq r.$

\end{definition}

\begin{proposition} \label{prop:interleaving}\label{lemma:interleaving}
(i) If  $\mathfrak{U}$ and $\mathfrak{V}$ are $\eta_1$-interleaved and $\mathfrak{V}$ and $\mathfrak{W}$ are $\eta_2$-interleaved, then, $\mathfrak{U}$ and $\mathfrak{W}$ are $(\eta_1+\eta_2)$-interleaved. 
(ii) Let $f:X\rightarrow Z$ be a continuous function and $\mathfrak{U}$ and 
$\mathfrak{V}$ be two $\eta$-interleaved tower of covers 
of $Z$. Then, $f^\ast(\mathfrak{U})$ and 
$f^\ast(\mathfrak{V})$ are also $\eta$-interleaved.
\end{proposition}

\begin{definition}[Interleaving of simplicial towers] \label{def:inter-cover}
Let $\mathfrak{S}=\big\{\mathcal{S}_{\varepsilon}\overset{\tiny{s_{\varepsilon,\varepsilon'}}}{\longrightarrow}\mathcal{S}_{\varepsilon'}\big\}_{r\leq \varepsilon\leq\varepsilon'}$ and $\mathfrak{T}=\big\{\mathcal{T}_{\varepsilon}\overset{\tiny{t_{\varepsilon,\varepsilon'}}}{\longrightarrow}\mathcal{T}_{\varepsilon'}\big\}_{r\leq\varepsilon\leq\varepsilon'}$ be two towers of simplicial complexes where
$\res(\mathfrak{S})=\res(\mathfrak{T})=r$. We say that they are $\eta\geq 0$ interleaved if for each $\varepsilon\geq r$ one can find simplicial maps $\varphi_\varepsilon:\mathcal{S}_\varepsilon \rightarrow \mathcal{T}_{\varepsilon+\eta}$ and  $\psi_\varepsilon:\mathcal{T}_\varepsilon \rightarrow \mathcal{S}_{\varepsilon+\eta}$ so that:
\begin{itemize}

\item[(i)]  for all $\varepsilon\geq r$, $\psi_{\varepsilon+\eta}\circ\varphi_{\varepsilon}$ and $s_{\varepsilon,\varepsilon+2\eta}$ are contiguous,

\item[(ii)]  for all $\varepsilon\geq r$, $\varphi_{\varepsilon+\eta}\circ\psi_{\varepsilon}$ and  $t_{\varepsilon,\varepsilon+2\eta}$ are contiguous,

\item[(iii)]  for all $\varepsilon'\geq\varepsilon\geq r$, $\varphi_{\varepsilon'}\circ s_{\varepsilon,\varepsilon'}$ and  $t_{\varepsilon+\eta,\varepsilon'+\eta}\circ \varphi_{\varepsilon}$ are contiguous,

\item[(iv)]  for all $\varepsilon'\geq\varepsilon\geq r$, $ s_{\varepsilon+\eta,\varepsilon'+\eta}\circ \psi_{\varepsilon}$ and  $\psi_{\varepsilon'}\circ t_{\varepsilon,\varepsilon'}$ are contiguous.
\end{itemize}

These four conditions are  summarized by requiring that the four diagrams below commute up to contiguity:

 \begin{align}\label{eq:simp-inter}
\xymatrix{
\mathcal{S}_{\varepsilon} \ar[dr]^{\varphi_\varepsilon} \ar[rr]^{s_{\varepsilon,\varepsilon+2\eta}} && \mathcal{S}_{\varepsilon+2\eta}\\
&\mathcal{T}_{\varepsilon+\eta} \ar[ur]^{\psi_{\varepsilon+\eta}}&
 } &
 \xymatrix{
&\mathcal{S}_{\varepsilon+\eta} \ar[rd]^{\varphi_{\varepsilon+\eta}}&\\
\mathcal{T}_{\varepsilon}\ar[ur]^{\psi_\varepsilon}\ar[rr]^{t_{\varepsilon,\varepsilon+2\eta}} & &\mathcal{T}_{\varepsilon+2\eta}
 } 
\end{align}
\begin{equation*}
\xymatrix{
\mathcal{S}_{\varepsilon} \ar[dr]^{\varphi_\varepsilon} \ar[rr]^{s_{\varepsilon,\varepsilon'}} && \mathcal{S}_{\varepsilon'}\ar[dr]^{\varphi_{\varepsilon'}}&\\
&\mathcal{T}_{\varepsilon+\eta}\ar[rr]_{t_{\varepsilon+\eta,\varepsilon'+\eta}} &&\mathcal{T}_{\varepsilon'+\eta}
 }
\end{equation*}
\begin{equation*}
\xymatrix{
&\mathcal{S}_{\varepsilon+\eta} \ar[rr]^{s_{\varepsilon+\eta,\varepsilon'+\eta}} && \mathcal{S}_{\varepsilon'+\eta}\\ \mathcal{T}_{\varepsilon}\ar[ur]^{\psi_{\varepsilon}}\ar[rr]_{t_{\varepsilon,\varepsilon'}} &&\mathcal{T}_{\varepsilon'}\ar[ur]^{\psi_{\varepsilon'}}
 }
\end{equation*}

\end{definition}

\begin{proposition}\label{prop:inter-pullback}
Let $\mathfrak{U}$ and $\mathfrak{V}$ be two $\eta$-interleaved towers of covers of $X$ with $\res(\mathfrak{U})=\res(\mathfrak{V})$. Then, $N(\mathfrak{U})$ and $N(\mathfrak{V})$ are also $\eta$-interleaved.
\end{proposition}

\begin{proof}
Let $r$ denote the common resolution of $\mathfrak{U}$ and $\mathfrak{V}$. Write $\mathfrak{U} = \big\{\mathcal{U}_\varepsilon\overset{\tiny{u_{\varepsilon,\varepsilon'}}}{\longrightarrow} \mathcal{U}_{\varepsilon'}\big\}_{r\leq\varepsilon\leq \varepsilon'}$ and $\mathfrak{V} = \big\{\mathcal{V}_\varepsilon\overset{\tiny{v_{\varepsilon,\varepsilon'}}}{\longrightarrow} \mathcal{V}_{\varepsilon'}\big\}_{r\leq\varepsilon\leq \varepsilon'}$, and for each $\varepsilon\geq r$ let $\zeta_\varepsilon:\mathcal{U}_\varepsilon\rightarrow \mathcal{V}_{\varepsilon+\eta}$ and  $\xi_\varepsilon:\mathcal{V}_\varepsilon\rightarrow \mathcal{U}_{\varepsilon+\eta}$ be given as in Definition \ref{def:inter-cover}. For each diagram in (\ref{eq:simp-inter}) one can consider a similar diagram at the level of covers involving covers of the form $\mathcal{U}_\varepsilon$ and $\mathcal{V}_\varepsilon$, and apply the nerve construction. This operation will yield diagrams identical to those in  (\ref{eq:simp-inter}) where for each $\varepsilon\geq r$:
\begin{itemize}
\item $\mathcal{S}_\varepsilon:=N(\mathcal{U}_\varepsilon)$,
$\mathcal{T}_\varepsilon:=N(\mathcal{V}_\varepsilon)$,
\item $s_{\varepsilon,\varepsilon'}:=N(u_{\varepsilon,\varepsilon'})$, for $r\leq \varepsilon\leq \varepsilon'$;
 $t_{\varepsilon,\varepsilon'}:=N(v_{\varepsilon,\varepsilon'})$, for $r\leq \varepsilon\leq \varepsilon'$;
$\varphi_{\varepsilon}:=N(\zeta_{\varepsilon})$, and
$\psi_{\varepsilon}:=N(\xi_{\varepsilon})$.
\end{itemize}
To satisfy Definition \ref{def:inter-cover}, it remains to verify conditions (i) to (iv). We only verify (i), since the proof of the others follows the same arguments. For this, notice that both the composite map $\xi_{\varepsilon+\eta}\circ\zeta_\varepsilon$ and $u_{\varepsilon,\varepsilon+2\eta}$ are  maps of covers  from $\mathcal{U}_\varepsilon$ to $\mathcal{U}_{\varepsilon+2\eta}.$ By Lemma \ref{lemma:cover-contiguity} we then have that $N(\xi_{\varepsilon+\eta}\circ\zeta_\varepsilon)$ and $N(u_{\varepsilon,\varepsilon+2\eta}) = s_{\varepsilon,\varepsilon+2\eta}$ are contiguous. But, by the properties of the nerve construction $N(\xi_{\varepsilon+\eta}\circ\zeta_\varepsilon) = N(\xi_{\varepsilon+\eta})\circ N(\zeta_\varepsilon) = \psi_{\varepsilon+\eta}\circ\varphi_{\varepsilon}$, which completes the claim.
\end{proof}

From now on, for a finite tower of simplicial complexes $\mathfrak{S}$ and $k\in\N$, we denote by $\mathrm{D}_k\mathfrak{S}$ the $k$-th persistence diagram of $\mathfrak{S}$ with coefficients in a fixed field. Notice that,
since contiguous maps induce identity maps at the homology level,
applying the simplicial homology functor with coefficients in a field to diagrams such as those in (\ref{eq:simp-inter}) yields two persistence modules \emph{strongly} interleaved in the sense of \cite{CCGGO09}. Thus, we have a stability result for $\mathrm{D}_k\MM(\mathfrak{U},f)$ when $f$ is kept fixed but the 
tower of covers $\mathfrak{U}$ is perturbed. 

\begin{corollary}\label{coro:stab-cov}
For $\eta \geq 0$, let $\mathfrak{U}$ and $\mathfrak{V}$ be 
two finite towers of covers of $Z$ with $\res(\mathfrak{U})=\res(\mathfrak{V})>0$. Let $f:X\rightarrow Z$
be well-behaved and $\mathfrak{U}$ and $\mathfrak{V}$ be $\eta$-interleaved. Then, $\mathrm{MM}(\mathfrak{U},f)$ and $\mathrm{MM}(\mathfrak{V},f)$ are $\eta$-interleaved. In particular, 
the bottleneck distance between the persistence diagrams $\mathrm{D}_k\mathrm{MM}(\mathfrak{U},f)$ and  $\mathrm{D}_k\mathrm{MM}(\mathfrak{V},f)$ is 
at most $\eta$ for all $k\in\N$.
\end{corollary}


\subsection{Stability of Multiscale Mapper}
\label{subsec:stability}

The \emph{fixed} simplicial complex produced by the standard mapper may not admit a simple notion of stability. We elaborate this point by an example in Appendix \ref{appendix:instability-mapper}. 
In this section, we show that the multiscale  mapper on the other hand exhibits a stability property against perturbations both in functions and in tower of covers (the latter of which already discussed in Corollary \ref{coro:stab-cov}). 
However, to guarantee stability against changes in functions, it is necessary to restrict the multiscale mapper to  a special class of towers of covers, called \emph{(c,s)-good tower of covers}. We justify the necessity of considering such ($c,s$)-good tower of covers in Appendix \ref{subsec:badcover}. 

In what follows, we assume that the target compact topological space $Z$ is endowed with a metric $d_Z$. For a subset $O\subset Z$, let $\diam(O)$ denote its \emph{diameter}, that is, the number $\sup_{z,z'\in O}d_Z(z,z')$. For $\delta\geq 0$, let $O^\delta$ denote the set $\{z\in Z|\, d_Z(z,O)\leq \delta\}.$

\begin{definition}[Good tower of covers]\label{def:good-cover}
Let $c\geq 1$ and $s> 0$. We say that a finite tower of covers $\mathfrak{W}=\{\mathcal{W}_{\varepsilon}\overset{\tiny{w_{\varepsilon,\varepsilon'}}}{\longrightarrow} \mathcal{W}_{\varepsilon'}\big\}_{\varepsilon\leq \varepsilon'}$ for the compact metric space $(Z,d_Z)$ is $(c,s)$-good if: 
\begin{enumerate}\denselist
\item $\res(\mathfrak{W})=s$, and $s\leq \diam(Z)$;
\item $\diam(W)\leq \varepsilon$ for all $W\in \mathcal{W}_\varepsilon$ and all $\varepsilon\geq s$; and 
\item for all $O\subset Z$ with $\diam(O)\geq s$,
there exists $W\in \mathcal{W}_{c \cdot\diam(O)}$ such that $W\supseteq O$.
\footnote{This condition is related to the concept of Lebesgue number 
for covers of metric spaces where metric balls instead of sets with
bounded diameter are used.} 
\end{enumerate}
\label{def:goodcover-II}
\end{definition}

Intuitively, in a ($c,s$)-good tower of covers, the parameter $\varepsilon$ of a cover $\mathcal{W}_\varepsilon \in \mathfrak{W}$ is a scale parameter, which specifies an upper-bound for the \emph{size / resolution} of cover elements in $\mathcal{W}_\varepsilon$. Condition-3 in the above definition requires that any set $O \subset Z$ should be covered by an element from a cover of $\mathfrak{W}$ whose resolution is comparable to the size (diameter) of $O$. 
See Appendix \ref{appendix:subsec:goodcover} for  more discussions on the ($c,s$)-good tower of covers and its properties. 
In particular, in Appendix \ref{subsec:badcover}, we how that such a tower of covers is somewhat necessary to obtain a stability result. The ($c,s$)-good tower of covers is also computationally feasible. In Appendix \ref{appendix:sec:goodcovers} we provide an example of a construction of a tower of covers for compact metric spaces with bounded doubling dimension using $\varepsilon$-nets.

A priori the conditions defining a $(c,s)$-good tower of covers do not guarantee that the sets $O$ in $Z$ with diameter smaller than the resolution $s$ can be covered by some element $W$ of some $\mathcal{W}_\varepsilon$, for some $\varepsilon\geq s.$ The following proposition deals with this situation and is used later.

\begin{proposition}\label{prop:useful}
If $Z$ is path connected and $\mathfrak{U}$ is a $(c,s)$-good tower of covers of $Z$, then whenever $O\subset Z$ is s.t. $\diam(O)<s$, there exists $W\in \mathcal{W}_{c\cdot(\diam(O)+2s)}$ s.t. $W\supseteq O$. 
\end{proposition}

\begin{remark}
Note that for a small set $O$, $\diam(O)<s$, one obtains only that $O\subset W$ for some $W$ in $\mathcal{W}_{c\cdot(\diam(O)+2s)}$. This is in contrast with what property 3 above guarantees for sets with diameter $\diam(O)\geq s$: the existence of $W\in \mathcal{W}_{c\cdot\diam(O)}$ with $O\subset W.$
\end{remark}

\begin{proof}
Indeed, since $Z$ is path connected then one has that 
$\diam(O^s)\geq \min\{\diam(Z),s\}=s$ it follows by the definition that $\exists W'\in \mathcal{W}_{c \cdot \diam(O^s)}$ containing $O$. We conclude since $\diam(O^s)\leq \diam(O)+2s$, $O\subset O^s$, and since there exists $W\in\mathcal{W}_{c(\diam(O)+2s)}$ containing $W'$.\end{proof}

We will henceforth assume that $(Z,d_Z)$ is a compact path connected metric space.

Next we prove the main stability result, Theorem \ref{thm:PDapprox}. Specifically, in Section \ref{appendix:sec:stab-func} we first inspect the impact of function perturbation, and we give the general result in Section \ref{appendix:sec:stab-gen}.

\subsection{Stability against function perturbation}
\label{appendix:sec:stab-func}

Our study of stability against function perturbation involves
a reindexing of the involved towers of covers.
\begin{definition}[Reindexing]
Let $\mathfrak{W}=\big\{\mathcal{W}_\varepsilon\stackrel{w_{\varepsilon,\varepsilon'}}{\longrightarrow}\mathcal{W}_{\varepsilon'}\big\}_{r\leq \varepsilon\leq \varepsilon'}$ be a tower of covers of $Z$ with
$r=\res(\mathfrak{W})>0$ and $\phi:\R^+\rightarrow \R^+$ be a monotonically increasing function. Consider the tower of covers
$\mathrm{R}_\phi (\mathfrak{W}) :=\big\{\mathcal{V}_\varepsilon\stackrel{v_{\varepsilon,\varepsilon'}}{\longrightarrow}\mathcal{V}_{\varepsilon'}\big\}_{\phi(r)\leq \varepsilon\leq \varepsilon'}$ given by
\begin{itemize}
\item  $\mathcal{V}_\varepsilon:=\mathcal{W}_{\phi^{-1}(\varepsilon)}$ for each $\varepsilon\geq \phi(r)$ and
\item $v_{\varepsilon,\varepsilon'} = w_{\phi^{-1}(\varepsilon),\phi^{-1}(\varepsilon')}$ for $\varepsilon,\varepsilon'$ such that $\phi(r)\leq \varepsilon\leq\varepsilon'$.
\end{itemize}
We refer to $\mathrm{R}_\phi(\mathfrak{W})$ as a \emph{reindexed} tower.
\end{definition}

In our case, we will use the $\log$ function to
reindex a tower of covers $\mathfrak{W}$.
We also need the following definition in order to state the stability results.
\begin{definition}
Given a tower of covers $\mathfrak{U} = \{\mathcal{U}_\varepsilon\}$ and $\eps_0\geq \res(\mathfrak{U})$, we define the
\mbox{$\eps_0$-truncation} of $\mathfrak{U}$ as the tower $\trunc_{\varepsilon_0}(\mathfrak{U}):=\big\{\mathcal{U}_\eps\big\}_{\eps_0\leq\eps}$.
Observe that, by definition $\res(\trunc_{\varepsilon_0}(\mathfrak{U}))=\eps_0$.
\end{definition}

\begin{proposition}\label{prop:stab-func} Let $X$ be a compact topological space, $(Z,d_Z)$ be a compact path connected metric space, and $f,g:X\rightarrow Z$ be two continuous functions such that for some $\delta\geq 0$ one has that $\max_{x\in X}d_Z(f(x),g(x))= \delta.$
Let $\mathfrak{W}$ be any $(c,s)$-good
tower of covers of $Z$.
Let $\eps_0=\max(1,s)$. Then,
the $\log\eps_0$-truncations of
$\mathrm{R}_{\log} \big(f^\ast(\mathfrak{W})\big)$ and $\mathrm{R}_{\log} \big(g^\ast(\mathfrak{W})\big)$ are
$\log\big(2c\max(\delta,s)+c\big)$-interleaved.
\end{proposition}
\begin{proof}
For notational convenience write $\eta:=\log\big(2c\max(\delta,s)+c\big)$, $\{\mathcal{U}_t\}=\mathfrak{U}:=f^\ast(\mathfrak{W})$, and $\{\mathcal{V}_t\}=\mathfrak{V}:=g^\ast(\mathfrak{W})$. With regards to satisfying Definition \ref{def:inter-cover} for $\mathrm{R}_{\log}\big(\mathfrak{U}\big)$ and $\mathrm{R}_{\log}\big(\mathfrak{V}\big)$, for each $\varepsilon\geq \log\eps_0$ we need only
exhibit maps of covers
$\zeta_{\varepsilon}:\mathcal{U}_{\exp(\varepsilon)} \rightarrow
\mathcal{V}_{\exp(\varepsilon+\eta)}$ and
$\xi_{\varepsilon}:{\mathcal{V}}_{\exp(\varepsilon)} \rightarrow \mathcal{U}_{\exp(\varepsilon+\eta)}$.  We first establish the following.
\begin{claim}\label{claim-1}
For all $O\subset Z$, and all $\delta'\geq \delta$, $f^{-1}(O)\subseteq g^{-1}(O^{\delta'}).$
\end{claim}
\begin{proof}
Let $x\in f^{-1}(O)$, then $d_Z(f(x),O)=0$.   Thus, 
$$d_Z(g(x),O)\leq d_Z(f(x),O)+d_Z(g(x),f(x))\leq
\delta,$$ which implies the claim.
\end{proof}

Now, pick any $t\geq \varepsilon_0$, any $U\in \mathcal{U}_{t}$, and fix $\delta':= \max(\delta,s)$. Then, there exists $W\in\mathcal{W}_{t}$ such that $U\in \mathrm{cc}(f^{-1}(W)).$ The claim implies that $f^{-1}(W)\subseteq g^{-1}(W^{\delta'})$. Since $\mathfrak{W}$ is a $(c,s)$-good cover of the connected space $Z$ and  $s\leq \max(\delta,s)\leq \diam(W^{\delta'})\leq 2\delta' +t$, there exists at least one set $W'\in\mathcal{W}_{c(2\delta'+t)}$ such that $W^{\delta'}\subseteq W'$.
This means that $U$ is contained in some element of $\mathrm{cc}(g^{-1}(W'))$
where $W'\in \mathcal{W}_{c(2\delta'+t)}$. But, also, since
$c(2\delta'+t)\leq c(2\delta'+1)t$ for $t\geq\eps_0\geq 1$, there exists $W''\in \mathcal{W}_{c(2\delta' +1)t}$ such that $W'\subseteq W''$.
This implies that $U$ is contained in some element of $\mathrm{cc}(g^{-1}(W''))$
where $W''\in \mathcal{W}_{c(2\delta'+1)t}$.
This process, when applied to all $U\in\mathcal{U}_{t}$, all $t\geq \varepsilon_0$,
defines a map of covers $\widehat{\zeta}_t:\mathcal{U}_t\rightarrow \mathcal{V}_{(2c\delta'+c)t}$. Now, define for each $\varepsilon\geq \log(\varepsilon_0)$ the map $\zeta_{\varepsilon}:=\widehat{\zeta}_{\exp(\varepsilon)}$ and notice that by construction this map has $\mathcal{U}_{\exp(\varepsilon)}$ as domain, and $\mathcal{V}_{(2\delta' c+c)\exp(\varepsilon)}$ as codomain.
A similar observation produces for each $\varepsilon\geq \log(\varepsilon_0)$ a map of covers $\xi_{\varepsilon}$ from $\mathcal{V}_{\exp(\varepsilon)}$ to $\mathcal{V}_{(2c\delta'+c)\exp(\varepsilon)}$.

Notice that for each $\varepsilon\geq \log(\varepsilon_0)$ one may write $(2c\delta' + c)\exp(\varepsilon) = \exp(\varepsilon+\eta)$. So we have in fact proved that $\log\eps_0$-truncations of
$\mathrm{R}_{\log}\big(\mathfrak{U}\big)$ and $\mathrm{R}_{\log}\big(\mathfrak{V}\big)$ are
$\eta$-interleaved.
\end{proof}

An application of Proposition \ref{prop:stab-func} and the result
in~\cite{CCGGO09} is:
\begin{corollary}\label{coro:stab-func}
Let $\mathfrak{W}$ be any $(c,s)$-good tower of covers of the compact connected metric space $Z$ and let $f,g:X\rightarrow Z$ be any two
well-behaved continuous functions such that $\max_{x\in X}d_Z(f(x),g(x))= \delta$.
Then, the bottleneck distance between the persistence diagrams $\mathrm{D}_k\mathrm{MM}(\mathrm{R}_{\log}\big(\mathfrak{W}),f\big)$ and  $\mathrm{D}_k\mathrm{MM}\big(\mathrm{R}_{\log}(\mathfrak{W}),g\big)$ is
at most $\log(2c\max(s,\delta)+c)+\max(0,\log\frac{1}{s})$ for all $k\in\N$.
\end{corollary}

\begin{proof}
We use the notation of Proposition \ref{prop:stab-func}. Let $\mathfrak{U}=f^*(\mathfrak{W})$ and $\mathfrak{V}=g^*(\mathfrak{W})$.
If $\max(1,s)=s$, then $\mathrm{R}_{\log}(\mathfrak{U})$ and $\mathrm{R}_{\log}(\mathfrak{V})$
are $\log(2c\max(s,\delta) + c)$-interleaved by Proposition~\ref{prop:stab-func} 
which gives a bound on the bottleneck distance of $\log(2c\max(s,\delta)+c)$ between the corresponding persistence
diagrams~\cite{CCGGO09}. In the case when $s<1$, the bottleneck distance remains the 
same only for the $0$-truncations of
$\mathrm{R}_{\log}(\mathfrak{U})$ and $\mathrm{R}_{\log}(\mathfrak{V})$. 
By shifting the starting
point of the two families to the left by at most $|\log s|$ can introduce
barcodes of lengths at most $\log{\frac{1}{s}}$ or can stretch 
the existing barcodes to the left by at most $\log{\frac{1}{s}}$ for the
respective persistence modules.
To see this, consider the persistence module below where $\eps_1=\log s$:
\begin{align*}
\mathrm{H}_k\big(N(f^\ast(\mathcal{U}_{\varepsilon_1}))\big)&\rightarrow \mathrm{H}_k\big(N(f^\ast(\mathcal{U}_{\varepsilon_2}))\big)\rightarrow\cdots
\cdots\rightarrow\mathrm{H}_k\big(N(f^\ast(\mathcal{U}_{0}))\big)&\rightarrow\cdots
\rightarrow \mathrm{H}_k\big(N(f^\ast(\mathcal{U}_{\varepsilon_n}))\big)
\end{align*}

A homology class born at any index in the range $[\log s ,0)$ 
either dies at or before the index $0$ or is mapped to a homology class
of $\mathrm{H}_k\big(N(f^\ast(\mathcal{U}_{0}))\big)$. In the first
case we have a bar code of length at most $|\log s|=\log\frac{1}{s}$.
In the second case, a bar code of the persistence module

\begin{equation*}
\mathrm{H}_k\big(N(f^\ast(\mathcal{U}_{\varepsilon_1=0}))\big)\rightarrow \cdots
\rightarrow \mathrm{H}_k\big(N(f^\ast(\mathcal{U}_{\varepsilon_n}))\big)
\end{equation*}

starting at index $\log 1=0$
gets stretched to the left by at most $|\log s|=\log\frac{1}{s}$. 
The same conclusion can be drawn for the persistence module induced
by $\mathrm{R}_{\log}(\mathfrak{V})$. 
Therefore the bottleneck distance between the respective
persistence diagrams changes by at most $\log\frac{1}{s}$.
\end{proof}

\begin{remark}
The proposition and corollary above are in appearance somehow not satisfactory: imagine that $f=g$, then $\delta=0$ but by invoking Proposition \ref{prop:stab-func} instead of a $0$-interleaving one obtains a $\big(\log(c(2s+1))+\max(0,\log\frac{1}{s})\big)$-interleaving. Nevertheless, the claim of the proposition is almost tight for $\delta > 0$: in fact, by modifying the example in \S\ref{subsec:badcover} for each $\delta>0$ we can find a space $X_\delta$, a $(c,s)$-good tower of covers $\mathfrak{U}$ of a closed interval $I\subset \R$,  and a pair of functions $f_\delta,g_\delta:X\rightarrow I$ such that $\|f_\delta-g_\delta\|_{\infty}=\delta$ but such that  the bottleneck distance between the persistence diagrams of $\MM(\mathfrak{U},f_\delta)$ and $\MM(\mathfrak{U},g_\delta)$ is at least 
 $\max \{ \log (cs), \log c, \log \frac{1}{s} \} = \Omega \big(\log(c(2s+1))+\max(0,\log\frac{1}{s})\big)$. 
\end{remark}

\subsection{Stability in general}
\label{appendix:sec:stab-gen}

We now consider the more general case which also allows changes in the tower of covers inducing the multiscale mapper.

\begin{theorem}\label{theo:stab-general}
Let $\mathfrak{U}$ and $\mathfrak{V}$ be any two $\eta$-interleaved
$(c,s)$-good towers of covers of the compact connected metric
space $Z$ and let $f,g:X\rightarrow Z$ be any two continuous well-behaved functions
such that $\max_{x\in X}d_Z(f(x),g(x))\leq \delta$. Then,
for $\eps_0=\max(1,s)$, the $\log\eps_0$-truncation of
$\mathrm{R}_{\log}\big(f^\ast(\mathfrak{U})\big)$ and $\mathrm{R}_{\log}\big(g^\ast(\mathfrak{V})\big)$ are
$\log\big(2c\max(s,\delta)+c+\eta\big)$-interleaved.
\end{theorem}
\begin{proof}
Write $\trunc_{\varepsilon_0}(\mathfrak{V}) = \{\mathcal{V}_\varepsilon\}_{\varepsilon\geq \varepsilon_0}$ and $\trunc_{\varepsilon_0}(\mathfrak{U}) = \{\mathcal{U}_\varepsilon\}_{\varepsilon\geq \varepsilon_0}$. Denote $f^\ast(\mathfrak{V})=\{\mathcal{V}_{\varepsilon}^f\}_{\varepsilon\geq \varepsilon_0}$,
$g^\ast(\mathfrak{V})=\{\mathcal{V}^g_{\varepsilon}\}_{\varepsilon\geq \varepsilon_0}$, and
$f^\ast(\mathfrak{U})=\{\mathcal{U}_{\varepsilon}^f\}_{\varepsilon\geq \varepsilon_0}$,
$g^\ast(\mathfrak{U})=\{\mathcal{U}_{\varepsilon}^g\}_{\varepsilon\geq \varepsilon_0}$.
By following the argument and using the notation in Proposition~\ref{prop:stab-func},
for each $\varepsilon\geq \varepsilon_0$ one can find maps of covers
$\mathcal{U}^g_\eps\rightarrow \mathcal{U}^f_{(2c\delta'+c)\eps}$ and
$\mathcal{V}^f_\eps\rightarrow \mathcal{V}^g_{(2c\delta'+c)\eps}$.
Also, since ${\mathfrak{U}}^f$ and
${\mathfrak{V}}^f$ are
$\eta$-interleaved, for each $\eps\geq \varepsilon_0$ there are maps of covers
${\mathcal{U}}^f_\eps\rightarrow \mathcal{V}^f_{\eps+\eta}$ and
$\mathcal{V}^g_\eps\rightarrow {\mathcal{U}}^g_{\eps+\eta}$. Then, for each $\varepsilon\geq \varepsilon_0$ one can form the following diagram
$$\mathcal{U}_\varepsilon^f\longrightarrow \mathcal{U}^g_{\varepsilon(2c\delta'+c)}\longrightarrow \mathcal{V}^g_{\varepsilon(2c\delta'+c)+\eta}\hookrightarrow \mathcal{V}^g_{\varepsilon(2c\delta'+c+\eta)},$$

where the last step follows because since $\varepsilon\geq \varepsilon_0\geq 1$, then  we have that $\varepsilon(2c\delta'+c)+\eta \leq \varepsilon(2c\delta'+c+\eta).$ Thus, by composing the maps intervening in the diagram above we have constructed for any $\varepsilon\geq \varepsilon_0$ a map of covers $\mathcal{U}_\varepsilon^f\longrightarrow \mathcal{V}^g_{\varepsilon(2c\delta'+c+\eta)}$. In a similar manner one can construct a map of covers $\mathcal{V}_\varepsilon^g\longrightarrow \mathcal{U}^f_{\varepsilon(2c\delta'+c+\eta)}$ for each $\varepsilon\geq\varepsilon_0$.

This provides the maps
$\mathcal{U}^f_{\exp(\eps)}\rightarrow {\mathcal{V}}^g_{(2c\delta'+c+\eta)\exp(\eps)}$ and
$\mathcal{V}^g_{\exp(\eps)}\rightarrow {\mathcal{U}}^f_{(2c\delta'+c+\eta)
\exp(\eps)}$ for each $\varepsilon\geq \varepsilon_0$. Since $(2c\delta'+c+\eta)\exp(\eps)=\exp(\eps+ \log (2c\delta'+c+\eta))$, by reindexing by $\log$ we obtain that
$\mathrm{R}_{\log}\big(f^\ast(\trunc_{\varepsilon_0}(\mathfrak{U}))\big)$ and
$\mathrm{R}_{\log}\big(g^\ast(\trunc_{\varepsilon_0}(\mathfrak{V}))\big)$ are $\log(2c\delta'+c+\eta)$-interleaved.
\end{proof}
Similar to the derivation of Corollary~\ref{coro:stab-func}
from Proposition~\ref{prop:stab-func}, we get the following result
from Theorem~\ref{theo:stab-general}.

\begin{theorem}
Let $\mathfrak{U}$ and $\mathfrak{V}$ be any two $\eta$-interleaved,
$(c,s)$-good tower of covers of the compact path connected metric space $Z$ and let $f,g:X\rightarrow Z$ be any two 
well-behaved continuous functions such that $\max_{x\in X}d_Z(f(x),g(x))\leq \delta$. Then, the bottleneck distance between $\mathrm{D}_k\mathrm{MM}\big(
\mathrm{R}_{\log}(\mathfrak{U}),f\big)$ and  $\mathrm{D}_k\mathrm{MM}\big(\mathrm{R}_{\log}(\mathfrak{V}),g\big)$ 
is bounded by $\log\big(2c\max(s,\delta)+c+\eta\big)+
\max(0,\log{\frac{1}{s}})$ for all $k\in\N$.
\label{thm:PDapprox}
\end{theorem}

\section{Exact Computation for PL-functions on simplicial domains}
\label{sec:exactcomp}
The stability result in Theorem~\ref{thm:PDapprox} further motivates us 
to design
efficient algorithms for constructing multiscale mapper or its approximation
in practice. We justify the approximation algorithms 
using the stability results discussed above.

One of the most common types of input in practice
is a real-valued piecewise-linear (PL) function $f: |K| \to \mathbb{R}$ defined on the underlying space $|K|$ of a simplicial complex $K$. That is, $f$ is given at the vertex set $\vrt(K)$ of $K$, and linearly interpolated within any other simplex $\sigma\in K$. 

In what follows, we consider this PL setting, and show that interestingly, if the input function satisfies a mild ``minimum diameter" condition, then we can compute both mapper and multiscale mapper from simply the 1-skeleton (graph structure) of $K$. This makes the computation of the multiscale mapper from a PL function significantly faster and simpler as its time complexity depends on the size of the 1-skeleton of $K$, which is typically orders of magnitude smaller than the total number of simplices (such as triangles, tetrahedra, etc) in $K$.

Given a simplicial complex $K$, let $K_1$ denote the 1-skeleton of $K$: that is, $K_1$ contains the set of vertices and edges of $K$. Define $\tilde f: |K_1| \rightarrow \mathbb{R}$ to be the restriction of $f$ to $|K_1|$; that is, $\tilde f$ is the PL function on $|K_1|$ induced by function values at vertices.

\begin{condition}[Minimum diameter condition]\label{cond:md}
For a given tower of covers $\mathfrak{W}$ of a compact connected metric space $(Z,d_Z)$, let $\kappa(\mathfrak{W}):=\inf\{\diam(W);\,W\in\mathcal{W}\in\mathfrak{W}\}$ 
denote the minimum diameter of any element of any cover of the tower $\mathfrak{W}$. 
Given a simplicial complex $K$ with a function
$f:|K|\rightarrow Z$ and a tower of covers $\mathfrak W$ of the metric space $Z$,
we say that $(K,f,{\mathfrak W})$ satisfies the {\em minimum diameter}
condition if $\diam(f(\sigma)) \leq 
\kappa({\mathfrak W})$ for every simplex $\sigma \in K$.
\end{condition}
In our case, $f$ is a PL function, and thus satisfying the minimum diameter 
condition means that for every edge 
$e = (u,v)\in K_1$, $|f(u) - f(v)| \leq \kappa(\mathfrak{W})$. 
In what follows we assume that $K$ is connected. We do not lose
any generality by this assumption because the arguments below
can be applied to each connected component of $K$.

\begin{definition}
Two towers of  simplicial complexes 
$\mathfrak{S}=\big\{\mathcal{S}_{\varepsilon}\overset{\tiny{s_{\varepsilon,\varepsilon'}}}{\longrightarrow}\mathcal{S}_{\varepsilon'}\big\}$
and
$\mathfrak{T}=\big\{\mathcal{T}_{\varepsilon}\overset{\tiny{t_{\varepsilon,\varepsilon'}}}{\longrightarrow}\mathcal{T}_{\varepsilon'}\big\}$  
are isomorphic, denoted $\mathfrak{S}\cong \mathfrak{T}$, if $\res(\mathfrak{S})=\res(\mathfrak{T})$, and there exist simplicial isomorphisms $\eta_\varepsilon$ and $\eta_{\varepsilon'}$ such that the diagram below commutes for all $\res(\mathfrak{S}) \le \varepsilon \le \varepsilon'$. 
\end{definition}

$$ \xymatrix{
\mathcal{S}_\varepsilon \ar[r]^{s_{\varepsilon,\varepsilon'}} \ar@{=}[d]_{\eta_\varepsilon} & \mathcal{S}_{\varepsilon'} \ar@{=}[d]_{\eta_{\varepsilon'}} \\
\mathcal{T}_\varepsilon \ar[r]^{t_{\varepsilon,\varepsilon'}} & \mathcal{T}_{\varepsilon'} 
}
$$

\noindent Our main result in this section is the following theorem which enables us to compute the mapper, multiscale mapper, as well as the persistence diagram for the multiscale  mapper of a PL function $f$ from its
restriction $\tilde{f}$ to the 1-skeleton of the 
respective simplicial complex.

\begin{theorem}\label{thm:exactcomp}
Given a PL function $f:|K|\rightarrow \mathbb{R}$ 
and a tower of covers $\mathfrak W$ of
the image of $f$ with
$(K,f,{\mathfrak W})$ satisfying the minimum diameter condition, 
one has $\MM(\mathfrak{W},f)\cong \MM(\mathfrak{W},\tilde{f})$.
\end{theorem}

We show in Proposition \ref{prop:samenerve} that the two mapper outputs
$\M({\mathcal W},f)$ and
$\M({\mathcal W},\tilde f)$ are identical up to a relabeling
of their vertices (hence simplicially isomorphic) for every
${\mathcal W}\in \mathfrak{W}$.
Also, since the simplicial maps in the filtrations
$\MM({\mathfrak W},f)$ and $\MM({\mathfrak W},{\tilde f})$
are induced by the pullback of the same tower of covers ${\mathfrak W}$,
they are identical again up to the same relabeling of the vertices.
This then establishes the theorem.

In what follows, for clarity of exposition, we use $X$ and $X_1$ to denote the underlying space $|K|$ and $|K_1|$ of $K$ and $K_1$, respectively.
Also, we do not distinguish between a simplex $\sigma\in K$ and its
image $|\sigma|\subseteq X$ and thus freely say $\sigma\subseteq X$
when it actually
means that $|\sigma|\subseteq X$ for a simplex $\sigma\in K$.

\begin{proposition} \label{prop:samenerve}
If $(K,f,\mathfrak{W})$ satisfies the minimum diameter condition,
then for every ${\mathcal W}\in\mathfrak{W}$,
$\M({\mathcal W},f)$ is identical to $\M({\mathcal W},\tilde f)$
up to relabeling of the vertices.
\end{proposition}
\begin{proof}
Let ${\mathcal U}=f^*(\mathcal W)$ and $\tilde {\mathcal U}={\tilde f}^*(\mathcal W)$.
By definition of $\tilde f$, each $\tilde U\in \tilde{\mathcal U}$
is a connected component of some $U\cap X_1$ for some $U\in \mathcal U$.
In Proposition~\ref{connected-prop}, we show that
$U\cap X_1$ is connected for every $U\in \mathcal U$.
Therefore, for every element $U\in \mathcal{U}$,
there is a unique element $\tilde U=U\cap X_1$ in $\tilde {\mathcal U}$
and vice versa.
Claim \ref{intersect-prop} finishes the proof. \end{proof}
\begin{claim}
$\bigcap_{i=1}^k U_i\not=\emptyset$ if and only if
$\bigcap_{i=1}^k {\tilde U}_i \not=\emptyset$.
\label{intersect-prop}
\end{claim}
\begin{proof}
Clearly, if $\bigcap_{i=1}^k {\tilde U}_i\not = \emptyset$,
then $\bigcap_{i=1}^k U_i\not=\emptyset$ because
$\tilde U_i \subseteq U_i$.
For the converse, assume $\bigcap_{i=1}^k U_i\neq \emptyset$ and
pick any $x\in \bigcap_{i=1}^k U_i$. Let $\sigma\subseteq X$ be a simplex
which contains $x$. Consider the level set $L=f^{-1}(f(x))$.
Since $f$ is PL, $L_\sigma=L\cap \sigma$
is connected and contains a point $y\in\sigma\cap X_1$.
Then, $y\in U_i$ for each $i\in\{1,\ldots,k\}$
because $x$ and $y$ are connected by a path in $U_i\cap\sigma$
and $f(x)=f(y)$. Therefore, $y\in \bigcap_i U_i$.
Since $y\in X_1$ it follows that $y\in \bigcap_i {\tilde U}_i$.
\end{proof}

\begin{proposition}
If $(X,f,\mathfrak{W})$ satisfies the minimum diameter condition,
then for every ${\mathcal W}\in \mathfrak W$
and every $U\in f^*(\mathcal W)$, the
set $U\cap X_1$ is connected.
\label{connected-prop}
\end{proposition}

\begin{proof}
Fix $U\in f^\ast(\mathcal{W})$. If $U\cap X_1$ is not connected, let $C_1,\ldots,C_k$ denote its $k\geq 2$
connected components. First, we show that each $C_i$
contains at least one vertex of $X_1$. Let $e=(u,v)$ be any
edge of $X_1$ that intersects $U$. If both ends $u$ and
$v$ lie outside $U$, then $|f(u)-f(v)|> |\max_U f -\min_U f|\geq \kappa(\mathfrak{W})$.
But, this violates the minimum diameter condition. Thus, at least
one vertex of $e$ is contained in $U$. It immediately follows that
$C_i$ contains at least one vertex of $X_1$.

Let $\Delta$ be the set of all simplices $\sigma\subseteq X$ so that
$\sigma\cap U\not=\emptyset$. Fix $\sigma\in \Delta$ and let  $x$ be any point in
$\sigma\cap U$.

\begin{claim}\label{claim:PLcase}
There exists a point $y$ in an edge of $\sigma$ so that
$f(x)=f(y)$.
\end{claim}
\begin{proof}
We first observe that
$\vrt^+(\sigma):=\{v\in\vrt(\sigma)|f(v)\geq f(x)\}$
is non-empty. Otherwise $f(v) < f(x)$ for
for all vertices $v\in\sigma$.
In that case $f(x')<f(x)$ for all points $x'\in\sigma$ contradicting that $x$ belongs to $\sigma$.
Now, if $\vrt^+(\sigma)=\vrt(\sigma)$, there is an edge
$e\subseteq \sigma$ so that $x\in e$ and taking
$y=x$ serves the purpose. If $\vrt^+(\sigma)\not=\vrt(\sigma)$,
there is a vertex $u\in \vrt(\sigma)$ so that $f(u)< f(x)$.
Taking a vertex $v$ from the non-empty set $\vrt^+(\sigma)$, we get
an edge $e=(u,v)$ where $f(u)<f(x)$ and $f(v)\geq f(x)$.
Since $f$ is PL, it follows that $e$ has a point $y$ so that
$f(y)=f(x)$.
\end{proof}

Since $\sigma$ contains an edge $e$ that is intersected by $U$,
it contains a vertex of $e$ that is contained in $U$. This means
every simplex $\sigma\in \Delta$ has a vertex contained
in $U$. For each $i=1,\ldots,k$ let $\Delta_i:=\{\sigma\subseteq X\,|\, \vrt(\sigma)\cap C_i\neq \emptyset.\}$ Since every simplex $\sigma\in\Delta$ has a vertex contained in $U$,
we have $\Delta=\bigcup_i \Delta_i$. We argue that the sets $\Delta_1,\ldots,\Delta_k$
are disjoint from each other.  Otherwise, there exist  $i\neq j$ and
a simplex $\sigma$ with a vertex $u$ in $\Delta_i$ and
another vertex $v$ in $\Delta_j$. Then, the edge $(u,v)$ must be
in $U$ because $f$ is PL. But, this contradicts that $C_i$ and $C_j$
are disjoint. This establishes that each $\Delta_i$ is disjoint from
each other and hence $\Delta$ is not connected
contradicting that $U$ is connected. Therefore, our initial assumption that
$U\cap X_1$ is disconnected is wrong.
\end{proof}

\subsection{Real-valued functions on triangulable topological spaces}
\label{appendix:subsec:PLapprox}
We note that, by using PL functions to approximate the continuous functions, the above result also leads to an approximation of the multiscale mapper for real-valued functions defined on triangulable topological spaces.

A PL function $f: |K| \to \mathbb{R}$ \emph{$\delta$-approximates} a continuous function $g: X \to \mathbb{R}$ defined on a topological space $X$ if there exists a homeomorphism
$h:X\rightarrow |K|$ such that for any point $y\in |K|$, we have that $|f(y) - g\circ h^{-1}(y)| \le \delta$.
The following result states that if a PL function $\delta$-approximates a continuous real-valued function on $X$, then the persistence diagrams induced by the respective multiscale mappers are also close. 
\begin{corollary}\label{cor:apprX}
Let  $f: |K| \to \mathbb{R}$ be a piecewise-linear function on $K$ that $\delta$-approximates a continuous function $g: X \to \mathbb{R}$. Let $\mathfrak W$ be a $(c,s)$-good tower of covers of $\mathbb{R}$. Then the bottleneck distance between the persistence diagrams $\mathrm{D}_k\mathrm{MM}\big(\mathrm{R}_{\log}(\mathfrak{W}),f\big)$ and  $\mathrm{D}_k\mathrm{MM}\big(\mathrm{R}_{\log}(\mathfrak{W}),g\big)$ is
at most $\log\big(2c\max(s,\delta)+c\big)+\max(0,\log\frac{1}{s})$ for all $k \in \mathbb{N}$.
\end{corollary}

\begin{proof}
Let $\tilde{g}: |K| \to \mathbb{R}$ denote the push forward of $g: X \to \mathbb{R}$ by the homeomorphism $h: X \to |K|$. By the definition of $\delta$-approximation, we know that $\|f - \tilde g\|_\infty \le \delta$. Hence by Corollary \ref{coro:stab-func}, the bottleneck distance between the persistence diagrams $\mathrm{D}_k\mathrm{MM}\big(\mathrm{R}_{\log}(\mathfrak{W}),f\big)$ and  $\mathrm{D}_k\mathrm{MM}\big(\mathrm{R}_{\log}(\mathfrak{W}),\tilde g)$ is
at most $\log(2c\max(s,\delta)+c)+\max(0,\log\frac{1}{s})$ for all $k \in \mathbb{N}$.

On the other hand, since $h$ is homeomorphism, it is easy to verify that the pullback of $\mathfrak{W}$ via $g$ and via $\tilde g$ induce isomorphic persistence modules: $\mathrm{H}_k(\mathrm{MM}(\mathrm{R}_{\log}(\mathfrak{W}),g))) \cong \mathrm{H}_k(\mathrm{MM}(\mathrm{R}_{\log}(\mathfrak{W}),\tilde g)))$. Combining this with the discussion in previous paragraph, the claim then follows.
\end{proof}

\section{Approximating multiscale mapper for general maps}
\label{sec:approxcomp}
\newcommand{\ac}			{\rho}
\newcommand{\newW}		{\mathfrak{W}} 
\newcommand{\dW}			{\mathfrak{W}} 
\newcommand{\Gnerve}		{\mathrm{M}^{K_1}}
\newcommand{\Nnerve}		{\mathrm{M}}
\newcommand{\MG}[1]		{{\mathrm{M}^{#1}}}
\newcommand{\MMG}[1]		{{\mathrm{MM}^{#1}}}

While results in the previous section concern real-valued PL-functions, we now provide a significant generalization for the case where $f$ maps the underlying space of $K$ into an arbitrary compact metric space $Z$. We present a ``combinatorial'' version of the (multiscale) mapper where each connected component of
a pullback $f^{-1}(W)$ for any cover $W$ in the cover of $Z$ 
consists of only vertices of $K$. Hence, the construction of the Nerve complex for this modified (multiscale) mapper is purely combinatorial, simpler, and more efficient to implement. But we lose the ``exactness'', that is, in contrast with the guarantees provided by Theorem \ref{thm:exactcomp}, the simpler combinatorial mapper 
only \emph{approximates} the actual multiscale mapper at the homology level.
Also, it requires a $(c,s)$-good tower of covers of $Z$. One more caveat is that the towers of simplicial complexes arising in this case do not interleave in the (strong) sense of Definition \ref{def:inter-cover} but in a weaker sense. This limitation worsens the approximation result by a factor of $3$. 

In what follows, $\mathrm{cc}(O)$ for a set $O$ denotes the \emph{set} of all path connected components of $O$.

Given a map $f: |K|\to Z$ defined on the underlying space $|K|$ of a simplicial complex $K$, to construct the mapper and multiscale mapper, one needs to compute the pullback cover $f^*(\mathcal{W})$ for a cover $\mathcal{W}$ of the compact metric space $Z$.
Specifically, for any $W\in\mathcal{W}$ one needs to compute the pre-image $f^{-1}(W)\subset |K|$ and shatter it into connected components. Even in the setting adopted in \S \ref{sec:exactcomp}, where we have a PL function $ \tilde{f}: |K_1|\to \mathbb{R}$ defined on the 1-skeleton $K_1$ of $K$, the connected components in $\cc(\tilde{f}^{-1}(W))$ may contain vertices, edges, and also \emph{partial edges}: say for an edge $e \in K_1$, its intersection $e_W = e \cap f^{-1}(W) \subseteq e$, that is, $f(e_W) = f(e) \cap W$, is a partial edge. See Figure \ref{fig:comb-cc} for an example.  In general for more complex maps, $\sigma \cap f^{-1}(W)$ for any $k$-simplex $\sigma$ may be partial triangles, tetrahedra, etc., which can be nuisance for computations. The combinatorial version of mapper and multiscale mapper sidesteps this problem by ensuring that each connected component in the pullback $f^{-1}(W)$ consists of \emph{only vertices of $K$}.

\begin{figure}

\centerline{\includegraphics[width=0.5\textwidth]{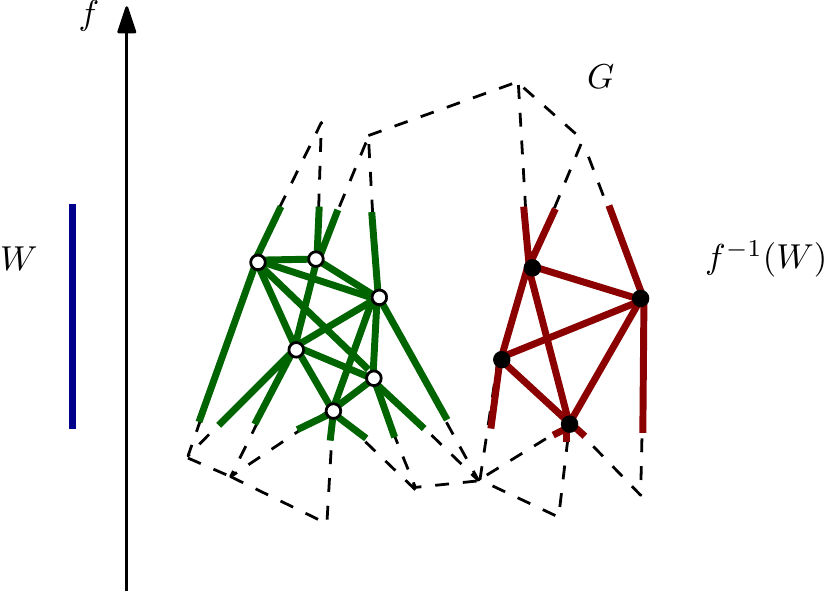}}

\caption{Partial thickened edges belong to the two connected components in $f^{-1}(W)$. Note that each set in $\cc_G(f^{-1}(W)$ contains only the set of vertices of a component in $\cc(f^{-1}(W))$.}
\label{fig:comb-cc}
\end{figure}

\subsection{Combinatorial mapper and multiscale mapper} 
Let $G$ be a graph with vertex set $\vrt(G)$ and edge set $\edg(G)$. Suppose we are given a map $f: \vrt(G) \to Z$ and a finite open cover $\mathcal{W} = \{W_\alpha \}_{\alpha\in A}$ of the metric space $(Z, d_Z)$. For any $W_\alpha \in \mathcal{W}$,  the pre-image $f^{-1}(W_\alpha)$ consists of a set of vertices 
which is shattered into
subsets by the connectivity of the graph $G$.
These subsets are taken as connected components.
We now formalize this:

\begin{definition}
Given a set of vertices $O \subseteq \vrt(G)$, \emph{the set of connected components of $O$ induced by $G$}, denoted by $\cc_G(O)$, is the partition of $O$ 
into a maximal subset of vertices connected in $G_O \subseteq G$, the subgraph spanned by vertices in $O$. We refer to each such maximal subset of vertices as a \emph{$G$-induced connected component of $O$}. We define $f^{\ast_G}(\mathcal{W})$, the \emph{$G$-induced pull-back via the function $f$}, as the collection of all $G$-induced connected components $\cc_G(f^{-1}(W_\alpha))$ for all $\alpha \in A$. 
\end{definition}

\begin{definition}($G$-induced multiscale mapper) \label{def:combMM}
Similar to the mapper construction, we define the \emph{$G$-induced mapper $\mathrm{M}^{G}(\mathcal{W},f)$} as the nerve complex $N(f^{*_{G}}(\mathcal{W}))$. 

Given a tower of covers  $\mathfrak{W}=\{\mathcal{W}_\varepsilon\}$ of $Z$, we define the \emph{$G$-induced multiscale mapper $\mathrm{MM}^G (\mathfrak{W},f)$} as the tower of $G$-induced nerve complexes $\{ N(f^{*_G}(\mathcal{W}_\eps)) \mid \mathcal{W}_\eps \in \mathfrak{W} \}$.
\end{definition}

\begin{algorithm}[t]
\caption{Combinatorial Multiscale Mapper} \label{alg:comb-mapper}
\begin{algorithmic}
\Require $f: |K| \to Z$ given by $f_V : V(K) \to Z$, a tower of covers $\mathfrak{W} = \{ \mathcal{W}_1, \ldots, \mathcal{W}_t \}$ 
\Ensure Persistence diagram $\mathrm{D}_*(\MM(\mathfrak{W}, f))$ 
\For{$i = 1, \ldots, t$}
\State{compute $V_W \subseteq V(K)$ where $f(V_W) = f(V(K)) \cap W$ and $\{V_W^j\}_j = \cc_{K_1}(V_W)$, for $\forall W \in \mathcal{W}_i$;}
\State{compute Nerve complex $N_i = N(\{V_W^j \}_{j, W} )$.}
\EndFor
\State{compute $\mathrm{D}_*( \{ N_i \!\to \!N_{i+1},\,{i\in [1, t-1]} \})$.}
\end{algorithmic}
\end{algorithm}

\subsection{Advantage of combinatorial multiscale mapper}
Given a map $f: |K| \to Z$ defined on the underlying space $|K|$ of a simplicial complex $K$, let $f_V: \vrt(K) \to \mathbb{R}$ denote the restriction of
$f$ to the vertices of $K$. Consider the graph $K_1$ as providing connectivity information for vertices in $\vrt(K)$. 
Given any tower of covers $\mathfrak{W}$ of the metric space $Z$,  
 the $K_1$-induced multiscale mapper $\mathrm{MM}^{K_1}(\mathfrak{W},f_V)$ is called the \emph{combinatorial multiscale mapper of $f$ w.r.t. $\mathfrak{W}$}.

A simple description of the computation of the combinatorial mapper is in Algorithm \ref{alg:comb-mapper}. For the simple PL example in Figure \ref{fig:comb-cc}, $f^{-1}(W)$ contains two connected components, one consists of the set of white dots, while the other consists of the set of black dots. More generally, the construction of the pullback cover needs to inspect only the 1-skeleton $K_1$ of $K$, which is typically of significantly smaller size. Furthermore, the construction of the Nerve complex $N_i$ as in Algorithm \ref{alg:comb-mapper} is also  much simpler: We simply remember, for each vertex $v\in V(K)$, the set $I_v$ of ids of connected components $\{V_W^j\}_{j, W \in \mathcal{W}_i}$ which contain it. Any subset of $I_v$ gives rise to a simplex in the Nerve complex $N_i$.

Let $\MM(\mathfrak{W},f)$ denote the standard multiscale mapper as introduced in \S\ref{sec:multimapper}. Our main result in this section is that if $\mathfrak{W}$ is a $(c,s)$-good tower of covers of $Z$, then the resulting two towers of simplicial complexes, $\MM(\mathfrak{W},f)$ and $\mathrm{MM}^{K_1}(\mathfrak{W},f_V)$, interleave in a sense which is weaker than that of Definition \ref{def:inter-cover} but still admits a bounded distance between their
respective persistence diagrams 
as a consequence of the weak interleaving result of \cite{CCGGO09}.
This weaker setting only worsens the approximation by a factor of $3$.

\begin{theorem}\label{thm:combinatorialrelation}
Assume that $(Z,d_Z)$ is a compact and connected metric space. Given a map $f: |K| \to Z$ defined on the underlying space of a simplicial complex $K$, let $f_V: \vrt(K) \to Z$ be the restriction of $f$ to the vertex set $\vrt(K)$ of $K$.  

Given a $(c,s)$-good tower of covers $\mathfrak{W}$ of $Z$ such that $(K,f,\mathfrak{W})$ satisfies the minimum diameter condition (cf. Condition \ref{cond:md}), 
the bottleneck distance between the persistence diagrams $\mathrm{D}_k\mathrm{MM}\big(\mathrm{R}_{\log}(\mathfrak{W}),f\big)$ and  
$\mathrm{D}_k\mathrm{MM}^{K_1}\big(\mathrm{R}_{\log}(\mathfrak{W}),f_V\big)$ is 
at most $3\log(3c)+3\max(0,\log\frac{1}{s})$ for all $k \in \mathbb{N}$. 
\end{theorem}

The remainder of this section is devoted to proving Theorem \ref{thm:combinatorialrelation}.

In what follows, the input tower of covers $\mathfrak{W} = \big\{ \mathcal{W}_{\varepsilon}\overset{\tiny{w_{\varepsilon,\varepsilon'}}}{\longrightarrow} \mathcal{W}_{\varepsilon'}\big\}_{\varepsilon\leq \varepsilon'}$ is $(c,s)$-good, and we set $\rho = 3c$. 

For each $\varepsilon\geq s$, let $\Nnerve_\eps$ denote the the nerve complex $N(f^* (\mathcal{W}_{\eps}))$ and $\Gnerve_\eps$ denote the the combinatorial mapper $N(f_V^{*_{K_1}}(\mathcal{W}_{\eps}))$.  Our goal is to show that there exist maps $\phi_{\varepsilon}$ and $\nu_\varepsilon$ so that diagram-(A) below commutes at the homology level, which then leads to a weak $\log \ac$-interleaving at log-scale of the persistence modules arising from computing the homologies of $\mathrm{MM}(\newW, f)$ and $\mathrm{MM}^{K_1}(\newW, f_V)$. This then proves  Theorem \ref{thm:combinatorialrelation}.

\begin{align}\label{eqn:nervemap}
\xymatrix{ &(A)& \\ 
\Nnerve_{\eps} \ar[drr]_{\phi_\eps} \ar[rr]^{\theta_{\eps, \rho\eps}} && \Nnerve_{\rho\eps}  \\
\Gnerve_{\eps} \ar[u]_{\nu_\eps}\ar[rr]_{\theta^{K_1}_{\eps,\rho\eps}} && \Gnerve_{\rho\eps} \ar[u]_{\nu_{\rho\eps}} 
}
&
\xymatrix{&(B)& \\ 
\Nnerve_{\eps}  \ar@{.>}[drr]_{\phi_\eps}\ar[rr]^{s_{\eps}} && \Nnerve_{\rho\eps}  \\
\Gnerve_{\eps} \ar[u]_{\nu_\eps}\ar[rr]_{t_{\eps}} && \Gnerve_{\rho\eps} \ar[u]_{\nu_{\rho\eps}} 
}
\end{align}

In the diagram above, the map $\theta_{\eps, \eps'} := N(f^* (w_{\eps, \eps'})): \Nnerve_\eps \to \Nnerve_{\eps'}$ is the simplicial map induced by the pullback of the map of covers $w_{\eps, \eps'}: \mathcal{W}_\eps \to \mathcal{W}_{\eps'}$. Similarly, the map $\theta^{K_1}_{\eps, \eps'} := N(f_V^{*_{K_1}} (w_{\eps, \eps'})): \Gnerve_\eps \to \Gnerve_{\eps'}$ is the simplicial map induced by the pullback of the map of covers $w_{\eps, \eps'}: \mathcal{W}_\eps \to \mathcal{W}_{\eps'}$.

\paragraph*{The maps $\nu_\eps$.}
Since connectivity in the 1-skeleton $K_1$ implies connectivity in the underlying space of $K$, there is a natural map from elements in $\cc_{K_1}(f_V^{-1}(W))$ to $\cc(f^{-1}(W))$, which in turn leads to a simplicial map $\nu_\eps: \Gnerve_\eps \to \Nnerve_\eps$.

The details are as follows. Given a nerve complex $\Nnerve_\eps$ (resp. $\Gnerve_\eps$), we abuse the notation slightly and identify a vertex in $\Nnerve_\eps$ with its corresponding connected component in $f^{*}(\mathcal{W}_{\eps})$ (resp. in $f_V^{*_{K_1}}(\mathcal{W}_{\eps})$). 

First, note that by construction of the combinatorial mapper, there is a natural map $u_W: \cc_{K_1}(f_V^{-1}(W)) \to \cc(f^{-1}(W))$ for any $W \in \mathcal{W}_\eps$: For any $V \in \cc_{K_1}(f_V^{-1}(W))$, obviously, all vertices in $V$ are $K_1$-connected in $f^{-1}(W)$. Hence there exists a unique set $U_V \in \cc(f^{-1}(W))$ such that $V \subseteq U_V$. We set $u_W(V) = U_V$ for each $V\in \cc_{K_1}(f_V^{-1}(W))$. 
The amalgamation of such maps for all $W \in \mathcal{W}_{\eps}$ gives rise to the map $\nu_\eps: f_V^{*_{K_1}}(\mathcal{W}_{\eps}) \to f^{*}(\mathcal{W}_{\eps})$, whose restriction to each $\cc_{K_1}(f_V^{-1}(W))$ is simply $u_W$ as defined above. 
Abusing notation slightly, we use $\nu_\eps: \vrt(\Gnerve_\eps) \to \vrt(\Nnerve_\eps)$ to denote the corresponding vertex map as well. 
Since for any set of vertices $V \in f_V^{*_{K_1}}(\mathcal{W}_{\eps})$, we have that $V \subseteq \nu_\eps(V) \subseteq |K|$. It then follows that non-empty intersections of sets in $f_V^{*_{K_1}}(\mathcal{W}_{\eps})$ imply non-empty intersections of their images via $\nu_\varepsilon$ in $f^{*}(\mathcal{W}_{\eps})$. Hence this vertex map induces a simplicial map which we still denote by $\nu_\eps: \Gnerve_{\eps}\to \Nnerve_\eps$.

\paragraph*{Auxiliary maps of covers $\mu_\eps$.}
What remains is to define the maps $\phi_\eps$ in diagram-(A). We first introduce the following map of covers:
$\mu_{\eps}: \mathcal{W}_\eps \to \mathcal{W}_{\rho\eps}$ for any $\eps \ge s$. 
Specifically, given any $W \in \mathcal{W}_{\eps}$ and $s\geq 0$, let $W^s = \{ x\in Z \mid d_Z (x, W) \le s \}$. 
Since $\mathfrak{W}$ is $(c,s)$-good, there exists at least one set $W' \in \mathcal{W}_{\rho\eps}$ such that $W^s \subseteq W'$. This is because for any $\eps \ge s$:  
$$s\leq\diam(W^s) \le \diam(W) + 2s =\eps + 2s \le 3\eps,$$
so that then $ c\cdot \diam(W^s) \le 3c \eps = \rho\eps. $
We set $\mu_{\eps}(W) = W'$: There may be multiple choices of sets in $\mathcal{W}_{\rho\eps}$ that contains $W$, we pick an arbitrary but fixed one. Let $s_{\eps}: \Nnerve_\eps \to \Nnerve_{\rho\eps}$ and $t_{\eps}: \Gnerve_\eps \to \Gnerve_{\rho\eps}$ denote the simplicial map induced by the pullbacks of the cover map $\mu_{\eps}: \mathcal{W}_\eps \to \mathcal{W}_{\rho\eps}$ via $f$ and via $f_V$, respectively. In other words, $s_{\eps} = N(f^*(\mu_{\eps}))$ and $t_{\eps} =N(f_V^{*_{K_1}}(\mu_{\eps}))$.

\paragraph*{The maps $\phi_\eps$.}
We now define the map $\phi_{\eps}: \Nnerve_{\eps} \to \Gnerve_{\rho\eps}$ with the help of diagram-(B) in (\ref{eqn:nervemap}). 
Fix $W \in \mathcal{W}_{\eps}$.  Given any set $U \in \cc(f^{-1}(W))$, write $\{ V_\beta \}_{\beta \in A_U}$ for the preimage $\nu_\eps^{-1}(U)$ of $U$ under the vertex map $\nu_\eps$.
 Note that $V_\beta \subseteq U$ for any $\beta\in A_U$ and $\bigcup_{\beta \in A_U} V_\beta = \vrt(K)\cap U.$

We claim that $t_{\eps}(V_\beta)  = t_{\eps} (V_{\beta'})$ for any $\beta, \beta' \in A_U$. Indeed, since $V_\beta$ and $V_{\beta'}$ are contained
in the path connected component $U$, let $\pi(x, y) \subseteq U$ be any path connecting a vertex $x\in V_\beta$ and $y\in V_{\beta'}$ in $U$. Let $\{\sigma_1, \ldots, \sigma_a\}$ be the collection of simplices that intersect $\pi(x,y)$. 
By the minimum diameter condition (cf. Condition \ref{cond:md}) and the definition of the map $\mu_\varepsilon$, we have that  $f(\sigma_j) \subseteq W^s \subseteq \mu_{\eps}(W) \in \mathcal{W}_{\rho\eps}, ~\text{for any}~ j\in \{1,\ldots, a\}.$

Hence $V_\beta \cup V_{\beta'} \cup \{\sigma_1, \ldots, \sigma_a\} \subseteq s_{\eps}(U) \in \cc(f^{-1}(\mu_\eps(W)))$. Furthermore, since the edges of simplices $\{\sigma_1, \ldots, \sigma_a\}$ connect $x\in V_\beta$ and $y\in V_{\beta'}$, vertices in $V_\beta$ and $V_{\beta'}$ are thus connected by edges in $s_{\eps}(U)$. That is, there exists a $K_1$-induced component (a subset of vertices of $K$) $\hat V \in \cc_{K_1} (f_V^{-1}(\mu_\eps(W))))$ such that $V_\beta \cup V_{\beta'} \subseteq \hat V$ for any $\beta, \beta' \in A_U$. 
Hence, all $V_\beta$s with $\beta\in A_U$ have the same image $\hat V$ in $\vrt(\Gnerve_{\rho\eps})$ under the map $t_{\eps}$, and we set $\phi_\eps(U):=\hat V$ to be this common image.

\begin{lemma}\label{lem:phieps}
The vertex map $\phi_\eps$ introduced above induces a simplicial map which we also denote by $\phi_\eps: \Nnerve_\eps \to \Gnerve_{\rho\eps}$. 
\end{lemma}
\begin{proof}
We have already proved in the main text that for all $\beta \in A_U$, $t_\eps (V_\beta)$ has the same image $\hat{V}$. 
Using the same notation as in the main text, recall that 
$$(\ast)~~ V_\beta \subseteq U ~\text{for any}~\beta\in A_U,~~~\text{and}~~~\bigcup_{\beta \in A_U} V_\beta = \vrt(K)\cap U. $$
This implies that $\vrt(K)\cap U \subseteq \phi_\eps(U)$. 
In fact, a similar argument can also be used to prove the following: 

\begin{claim}\label{claim:extendedvertices}
The vertices of any simplex that intersects $U$ will be contained in $\phi_\eps(U)$. 
\end{claim}

Now to prove Lemma \ref{lem:phieps}, we need to show the following: given a $k$-simplex $\tau = \{p_0, p_1, \ldots, p_k\} \in \Nnerve_\eps$, where each vertex $p_j$ corresponds to set $U_j \subseteq |K|$, then we have that $\bigcap_{j=0}^k \phi_\eps(U_j) \neq \emptyset$.  

To prove this, take any point $x\in |K|$ such that $x \in \bigcap_{j=0}^k U_j$. 

Suppose $x$ is contained in a  simplex $\sigma\in K$. 
By Claim \ref{claim:extendedvertices}, the vertices of $\sigma$ are contained in $\phi_\eps(U_j)$, that is, $\vrt(\sigma) \subset \phi_\eps(U_j)$, for any $j\in \{0,\ldots, k\}$. 
Hence $\bigcap_{j=0}^k \phi_\eps(U_j) \supseteq \vrt(\sigma) \neq \emptyset$. 
The lemma then follows. 
\end{proof}

Finally, by the construction of $\phi_\eps$, the lower triangle in diagram-(B) in (\ref{eqn:nervemap}) commutes. This fact, the commutativity of the
square in diagram-(B), and the definition of $\nu_\eps$ together 
imply that the top triangle commutes. 
Furthermore, Lemma \ref{lemma:cover-contiguity} states the two maps of covers $f^*( w_{\eps, \rho\eps})$ and $f^*(\mu_\eps)$ induce contiguous simplicial maps $\theta_{\eps,\rho \eps}$ and $s_{\eps}$. Similarly, $\theta^{K_1}_{\eps,\rho\eps}$ is contiguous to $t_{\eps}$. 
Since contiguous maps induce the same map at the homology level, it then follows that diagram-(A) in (\ref{eqn:nervemap}) commutes at the homology level. At this point one could apply the strong interleaving result of \cite{CCGGO09} had the maps  $\phi_\varepsilon$ been defined for any $ \mathrm{M}_\eps \to \mathrm{M}_{\eps'}^{K_1}$.  Nonetheless, the  Weak Stability Theorem of \cite{CCGGO09} can be applied even in this case. This, together with an argument similar to the proof of Corollary \ref{coro:stab-func}, provides Theorem \ref{thm:combinatorialrelation}.

\section{The metric point of view}
\label{sec:metricview}

The multiscale mapper construction from previous sections yields a tower of simplicial complexes, and 
a subsequent application of the homology functor with field coefficients provides a persistence module and hence a persistence diagram. 

An alternative idea is to use the data $(\mathfrak{U}, f:X\rightarrow Z)$  to \emph{induce a (pseudo-)metric on} $X$, and then consider the Rips or \v{C}ech filtrations arising from that metric, and in turn obtain a persistence diagram. Both viewpoints are valid in that they both produce topological summaries
out of the given data. 
The first approach proceeds at the level of filtrations and then simplicial maps, and then persistent homology, whereas the other readily produces a metric on $X$ and then follows the standard metric point of view with geometric complexes. Notice that these two approaches yield different computational problems: the $\MM$ approach leads to persistence under simplicial maps~\cite{DFW-sph}, and the metric approach leads to persistence under inclusion maps, which appears to be computationally easier with state-of-the-art techniques.

\subsection{The pull-back pseudo-metric}

Assume that a continuous function $f:X\rightarrow Z$ and a tower of covers $\mathfrak{U}$ with $\res(\mathfrak{U})\geq 0$ of $Z$ are given. 
We wish to interpret $\varepsilon$ as the size of each 
$U\in\mathcal{U}_\varepsilon$ at least for $\eps\geq s$ 
when $\mathfrak{U}$ is
a $(c,s)$-good tower of covers.  We define the function $d_{\mathfrak{U},f}:X\times X\rightarrow \R^+$ for $x,x'\in X$ s.t. $x\neq x'$ by
\begin{equation}
\label{eq:dUf} d_{\mathfrak{U},f}(x,x'):=\inf\big\{\varepsilon> 0|\,\exists V\in f^\ast(\mathcal{U}_\varepsilon)\,\,\mbox{with $x,x'\in V$}\big\}, ~\text{and} ~d_{\mathfrak{U},f}(x,x) := 0\,\,\mbox{for all $x\in X$}. 
\end{equation}
Notice that $d_{\mathfrak{U},f}(x,x')\geq s$ whenever $x\neq x'$. We refer to $d_{\mathfrak{U},f}$ as the \emph{pull-back pseudo-metric}.

\begin{lemma}\label{lemma:metric}
If $\mathfrak{U}$ is a $(c,s)$-good tower of covers of the compact connected metric space $Z$, then $d_{\mathfrak{U},f}$ satisfies:
\begin{itemize}

\item $d_{\mathfrak{U},f}(x,x') = d_{\mathfrak{U},f}(x',x)$ for all $x,x'\in X$.
\item $d_{\mathfrak{U},f}(x,x')\leq c\cdot\big(d_{\mathfrak{U},f}(x,x'')+d_{\mathfrak{U},f}(x'',x')+2s\big)$ for all $x,x',x''\in X$.
\end{itemize}
\end{lemma}

\begin{proof}[Proof of Lemma \ref{lemma:metric}]
That the definition yields a symmetric function is clear. The 
second claim follows
from the $c$-goodness of $\mathfrak{U}$ and Proposition \ref{prop:useful}. Indeed, 
assume that $V_1\in f^\ast(U_{\varepsilon_1})$ is s.t. $x,x''\in V_1$ and $V_2\in f^\ast(U_{\varepsilon_2})$ is s.t. $x'',x'\in V_2$ . Then there exists $U_1\in \mathcal{U}_{\varepsilon_1}$ s.t. $V_1\in\cc(f^{-1}(U_1)),$ and there exists $U_2\in \mathcal{U}_{\varepsilon_2}$ s.t. $V_2\in\cc(f^{-1}(U_2)).$ Clearly, $V_1\cap V_2\neq \emptyset$, implying $U_{1}\cap U_{2}\neq \emptyset$. Thus $\diam(U_1\cup U_2\big)\leq \varepsilon_1+\varepsilon_2$. The proof concludes via Proposition \ref{prop:useful}.
\end{proof}

\begin{remark}
Note that despite the fact that we call $d_{\mathfrak{U},f}$ a pseudo-metric, it does not satisfy the triangle inequality in a strict sense. According to the lemma above it does, however, satisfy a relaxed version of such inequality. Furthermore, in the ideal case when $c=1$ and $s=0$, $d_{\mathfrak{U},f}$ satisfies the triangle inequality.  Thus, in what follows we will take the liberty of calling it a pseudo-metric. Furthermore, this relaxed triangle inequality does not preclude our ability to consider the induced \v{C}ech complex and to develop the interleaving results in Section \ref{sec:inter-metric-mm}.

\end{remark}

\subsection{Stability of $d_{\mathfrak{U},f}$} In the same manner that we established the stability of the multiscale mapper construction, one can answer what is the stability picture for the \v{C}ech construction over $(X,d_{\mathfrak{U},f})$. A first step is to understand how $d_{\mathfrak{U},f}$ changes when we alter both the function $f$ and the tower $\mathfrak{U}.$

\begin{proposition}[Stability of pull-back pseudo-metric under cover perturbations]\label{prop:stab-pull-back-metric-cover}
Let $f:X\rightarrow Z$ be a continuous function and $\mathfrak{U}$ and $\mathfrak{V}$ be two $\eta$-interleaved towers of covers of $Z$ with
$\res(\mathfrak{U})>0$ and $\res(\mathfrak{V})>0$. Then, 

$$\forall x,x'\in X,\,\, d_{\mathfrak{U},f}(x,x')\leq d_{\mathfrak{V},f}(x,x')+ \eta.$$
\end{proposition}
\begin{proof}
Assume the hypothesis. Then, by Lemma \ref{lemma:interleaving} $f^\ast(\mathfrak{U})$ and $f^\ast(\mathfrak{V})$ are $\eta$-interleaved as well. Pick $x,x'\in X$ and assume that $d_{\mathfrak{U},f}(x,x')< \varepsilon$. Let $U\in f^{\ast}(\mathcal{U}_\varepsilon)$ be such that $x,x'\in U$.  Then, one can find $V\in \mathcal{V}_{\varepsilon+\eta}$ such that $U\subset V$. Then, since $x,x'\in V$, it follows that 
$$d_{\mathfrak{V},f}(x,x')\leq \varepsilon +\eta.$$
Since this holds for any $\varepsilon>d_{\mathfrak{U},f}(x,x')$, 
we obtain that $d_{\mathfrak{U},f}(x,x')\leq d_{\mathfrak{V},f}(x,x')+\eta.$  
\end{proof}

\begin{proposition}[Stability of pull-back pseudo-metric against function perturbation]\label{prop:stab-metric-func}
Let $\mathfrak{U}$ be a $(c,s)$-good tower of covers of the compact connected metric  space $Z$ and let  $f,g:X\rightarrow Z$ be two continuous functions such that for some $\delta\geq 0$ one has $\max_{x\in X}d_Z(f(x),g(x))\leq \delta$. Then, 
$$\forall x,x'\in X,\,\, d_{\mathfrak{U},g}(x,x')
\leq c\big(d_{\mathfrak{U},f}(x,x')+2\max(\delta,s)\big).$$
\end{proposition}
\begin{proof} 
Write $\mathfrak{U}=\{\mathcal{U}_\varepsilon\}$. Let $d_{\mathfrak{U},f}(x,x')<\eps$. 

Then there exists $U \in  \mathcal{U}_\eps$ and $V\in \cc (f^{-1}(U))$ such that $x, x' \in V$. 
Write $\delta'=\max(\delta,s).$
We know $f^{-1}(U)\subseteq g^{-1}(U^{\delta'})$. Since 
$s\leq \diam(U^{\delta'})\leq \eps+2\delta'$ and $\mathfrak{U}$ is a 
$(c,s)$-good tower of covers of $Z$, there exists 
$U'\in{\mathcal U}_{c(\eps+2\delta')}$ so that $U^{\delta'}\subset U'$. Thus, $V\subseteq g^{-1}(U')$.
It follows by definition that
$
d_{\mathfrak{U},g}(x,x')
\leq c\cdot (\eps+2\delta').
$
Since the argument holds for all $\eps>0$, we have the result.
\end{proof}

As a corollary of the two preceding propositions we obtain the following statement:
\begin{proposition}\label{prop:stab-metric-coro}
Let $\eps_0,c\geq 1$ and $\eta,s>0$. 
Let $\mathfrak{U}$ and $\mathfrak{V}$ be any two $\eta$-interleaved,
$\eps_0$-truncations of $(c,s)$-good towers of covers of the compact connected metric space $Z$ and let $f,g:X\rightarrow Z$ be any two continuous functions such that $\max_{x\in X}d_Z(f(x),g(x))\leq \delta$. Then, for all $x,x'\in X$
$\gamma^{-1}\cdot d_{\mathfrak{V},f}(x,x') \leq d_{\mathfrak{U},g}(x,x')\leq \gamma\cdot d_{\mathfrak{V},f}(x,x'),$ 
\label{metric-stability}
where $\gamma := \big(2c\max(s,\delta)+c+\eta\big)$.
\end{proposition}

\subsection{\v{C}ech filtrations using $d_{\mathfrak{U},f}$}

Given any point $x \in X$, we denote the $\varepsilon$-ball around $x$ by $B_\varepsilon(x) = \{ x' \in X \mid d_{\mathfrak{U},f}(x,x')  \le \eps \}$. We have the following observation. 

\begin{lemma}\label{lemma:ball-dUf}
Let $f:X\rightarrow Z$ be continuous and $\mathfrak{U}$ any 
tower of covers of $Z$ with $\res(\mathfrak{U})>0$.  Consider the pseudo-metric space $(X,d_{\mathfrak{U},f})$. Then, for any $x\in X$ and $\varepsilon\geq 0$, $$B_\varepsilon(x) = \bigcup_{V\in f^\ast(\mathcal{U}_\varepsilon),\,\mbox{s.t. $x\in V$}} V.$$ 
\end{lemma}
\begin{proof}
By definition, $B_\varepsilon(x) = \bigcup_{\delta\leq \varepsilon}\{V\in f^\ast(\mathcal{U}_\delta);\,x\in V\}.$
Now, since the tower $\mathfrak{U}$ is hierarchical (cf. Definition \ref{def:hfc}), then whenever $\delta\leq \varepsilon$ and $V\in f^\ast(\mathcal{U}_\delta)$ one will also have that there exists $V'\in f^\ast(\mathcal{U}_\varepsilon)$ with $V \subseteq V'$. The conclusion follows.
\end{proof}

The fact that as defined $d_{\mathfrak{U},f}$ is a non-negative symmetric function on $X\times X$ permits defining the \v{C}ech (and also Rips) filtration on the powerset $\powset(X)$ of $X$. 
The \emph{\v{C}ech filtration $\cech(\mathfrak{U},f)$ induced by data $(\mathfrak{U},f: X\to Z)$}, is given by the
function $F_C: \powset(X) \rightarrow \mathbb{R}$ where for each $\sigma\subset X$, $F_C(\sigma):=\inf_{x\in X}\sup_{x'\in\sigma}d_{\mathfrak{U},f}(x,x').$ 
Specifically, $\cech(\mathfrak{U},f)$ can be written as $\{\mC_\varepsilon\stackrel{\iota_{\varepsilon,\varepsilon'}}{\longrightarrow}\mC_{\varepsilon'}\}_{s\leq \varepsilon\leq\varepsilon'}$ where $\mC_\varepsilon := \{\sigma \mid F_C(\sigma) \le \varepsilon \}$, and $\iota_{\varepsilon,\varepsilon'}$ are the natural inclusion maps. Note that $\sigma = \{x_0, \ldots, x_k\} \in \mC_\varepsilon$ if and only if $\bigcap_{i} B_\varepsilon(x_i) \neq \emptyset$.

\begin{corollary}
Under the hypothesis of Theorem~\ref{metric-stability},
$\cech\big(\mathrm{R}_{\log}(\mathfrak{U}),f\big)$ and $\cech\big(\mathrm{R}_{\log}(\mathfrak{V}),g\big)$ are 
$\log\big(2c\max(s,\delta)+c+\eta\big)$-interleaved.
\end{corollary}

Because of the above corollary,
the persistence diagrams arising from the \v{C}ech filtration 
on $X$ are stable in the bottleneck distance.

\subsection{An interleaving between $\MM(\mathfrak{U},f)$ and $\cech(\mathfrak{U},f)$}\label{sec:inter-metric-mm}

Interestingly, it turns out that the two views: the pullback of covers and an induced pullback metric, are closely related. Specifically, by putting together the elements discussed in this section we can state a theorem specifying a comparability between $\MM(\mathfrak{U},f)$ and $\cech(\mathfrak{U},f)$ under the $\log$-reindexing.

\begin{theorem}\label{theo:inter-metric}
Let $(Z,d_Z)$ be a compact connected metric space, $\mathfrak{U}$ be a $(c,s)$-good tower of covers of $Z$ with $s\geq 1$, and $f:X\rightarrow Z$ be continuous. Then, the hierarchical families of simplicial complexes $\MM\big(\mathrm{R}_{\log}\mathfrak{U},f\big)$ and 
$\cech\big(\mathrm{R}_{\log}\mathfrak{U},f\big)$ are $\log(c(s+2))$-interleaved.
  \end{theorem}

\begin{corollary}\label{coro:comparability}
Under the hypotheses of the previous theorem, the persistence diagrams of $\MM\big(\mathrm{R}_{\log}\mathfrak{U},f\big)$ and 
$\cech\big(\mathrm{R}_{\log}\mathfrak{U},f\big)$ are at bottleneck distance bounded by $\log(c(s+2))$.
\end{corollary}

\begin{proof}[Proof of Theorem \ref{theo:inter-metric}]
To fix notation, write $\mathfrak{U}=\{\mathcal{U}_\varepsilon\},$ $\MM(\mathfrak{U},f)=\{\M_\varepsilon\stackrel{s_{\varepsilon,\varepsilon'}}{\longrightarrow} \M_{\varepsilon'}\}_{s\leq\varepsilon\leq\varepsilon'}$, and 
$\cech(\mathfrak{U},f)=\{\mC_\varepsilon\stackrel{\iota_{\varepsilon,\varepsilon'}}{\longrightarrow}\mC_{\varepsilon'}\}_{s\leq \varepsilon\leq\varepsilon'}$  for the multiscale mapper and \v{C}ech filtrations, respectively. Recall that  $\sigma = \{x_0, \ldots, x_k\} \in \mC_\varepsilon$ if and only if $\bigcap_{i} B_\varepsilon(x_i) \neq \emptyset$. 

For each $\varepsilon\geq s$ we will define simplicial maps $\pi_\varepsilon:\mC_\varepsilon\rightarrow \M_{\mycechconst \varepsilon}$ and $\psi_{\varepsilon}: \M_{\varepsilon}\rightarrow\mC_{\varepsilon}$ so that each of the diagrams below 
commutes up to contiguity:

\begin{equation}\label{eqn:cechMM-1}
\xymatrix{
\M_{\varepsilon} \ar[d]_{\psi_\varepsilon} \ar[rr]^{s_{\varepsilon,\mycechconst \varepsilon}} && \M_{\mycechconst \varepsilon} \\
 {\mC}_{\varepsilon} \ar[urr]^{\pi_{\varepsilon}}&& 
}
\end{equation}

\begin{equation}\label{eqn:cechMM-2}
\xymatrix{
&& \M_{\mycechconst \varepsilon} \ar[d]^{\psi_{\mycechconst \varepsilon}}  \\
 {\mC}_{\varepsilon} \ar[urr]^{\pi_{\varepsilon}}\ar[rr]^{\iota_{\varepsilon,\mycechconst \varepsilon}}&& {\mC}_{\mycechconst \varepsilon}
}
\end{equation}

\begin{equation}\label{eqn:cechMM-3}
\xymatrix{
&&\M_{\mycechconst \varepsilon} \ar[rrr]^{s_{\mycechconst \varepsilon,\mycechconst \varepsilon'}} &&& \M_{\mycechconst \varepsilon'} \\
 {\mC}_{\varepsilon} \ar[urr]^{\pi_{\varepsilon}}\ar[rrr]^{\iota_{\varepsilon,\varepsilon'}}&&& {\mC}_{\varepsilon'}\ar[urr]^{\pi_{\varepsilon'}}
}
\end{equation}

\begin{equation}\label{eqn:cechMM-4}
\xymatrix{
\M_{\varepsilon} \ar[d]_{\psi_\varepsilon} \ar[rr]^{s_{\varepsilon,\varepsilon'}} && \M_{\varepsilon'} \ar[d]^{\psi_{\varepsilon'}}  \\
 {\mC}_{\varepsilon} \ar[rr]^{\iota_{\varepsilon,\varepsilon'}}&& {\mC}_{\varepsilon'}
}
\end{equation}

Recall that the vertex set of $\mC_\varepsilon$ is $X$, whereas the vertex set of $\M_\varepsilon$ is the cover $f^\ast(\mathcal{U}_\varepsilon).$

\paragraph*{The map $\pi_\varepsilon:\mC_\varepsilon\rightarrow \M_{\mycechconst \varepsilon}$.}
Consider the  map $\widehat{\pi}_\eps: X\rightarrow f^\ast(\mathcal{U}_{\mycechconst \varepsilon})$, 
where $\widehat{x}$ equals to an arbitrary but fixed $V_x\in f^\ast(\mathcal{U}_{\mycechconst \varepsilon})$ such that $B_\varepsilon(x)\subseteq V_x.$ Such a $V_x$ always exists. Indeed, note first that $f(B_\varepsilon(x))=\{U\in\mathcal{U}_\varepsilon|\,f(x)\in U\}.$
Then, it follows that $\diam(f(B_\varepsilon(x)))\leq 2\varepsilon$. Now, invoking the fact that $\mathfrak{U}$ is $(c,s)$-good and Proposition \ref{prop:useful}, and noting that $c(2\varepsilon + s) \le \mycechconst \varepsilon$ for $\eps \ge s\ge 1$, we conclude the existence of $U\in \mathcal{U}_{\mycechconst \varepsilon}$ such that $f(B_\varepsilon(x))\subset U$. Finally, pick $V_x\in\cc(f^{-1}(U))$ s.t. $B_\varepsilon(x)\subset V_x$.

The map $\widehat{\pi}_\eps$ induces the simplicial map
$\pi_\eps:\mC_\varepsilon\rightarrow \M_{\mycechconst \varepsilon}$ for any $\varepsilon \ge s$, as used in diagrams (\ref{eqn:cechMM-1}), (\ref{eqn:cechMM-2}) and (\ref{eqn:cechMM-3}). Indeed, assume 
$\sigma=\{x_0,\ldots,x_k\}\in\mC_\varepsilon$. Then, 
$\bigcap_{i}B_\varepsilon(x_i)\neq \emptyset.$ Since by construction 
$\widehat{\pi}_\varepsilon(x_i)=V_{x_i}\supseteq B_\varepsilon(x_i)$ for all $i$, 
it follows that $\bigcap_i \widehat{\pi}_\varepsilon(x_i)\neq \emptyset$ as well.

\paragraph*{The map $\psi_\varepsilon:\M_\varepsilon\rightarrow \mC_{\varepsilon}$.}
For all $V\in f^\ast(\mathcal{U}_\varepsilon)$ pick a point $x_V\in V$: This is a choice of a \emph{representative} for each element of the pullback cover. Define $\widehat{\psi}_\varepsilon(V)=x_V$. We now check that this vertex map induces a simplicial map $\psi_\varepsilon$.
Assume that $V_0,V_1,\ldots,V_k\in f^\ast(\mathcal{U}_\varepsilon)$ are s.t. $\bigcap_i V_i\neq \emptyset$. Note that by Lemma \ref{lemma:ball-dUf}, each $V_i$ satisfies $V_i \subseteq B_\varepsilon(x_{V_i})$. It then follows that  
$\bigcap_i B_{\varepsilon}(x_{V_i})\neq \emptyset,$ implying that $x_{V_0}, \ldots, x_{V_k}$ span a simplex in $\mC_{\varepsilon}$.

\begin{claim}\label{prop:commute-1}
The maps $s_{\eps,c(s+2)\eps}$ and $\pi_\varepsilon\circ \psi_{\varepsilon}$ in (\ref{eqn:cechMM-1}) are contiguous. 
\end{claim}
\begin{proof}
Assume that $V_0, V_1, \ldots, V_k \in f^*(\mathcal{U}_\varepsilon)$ are such that $\cap_i V_i \neq \emptyset$, and let $U_i := s_{\eps,c(s+2)\eps}(V_i)$ for each $i$. Since $V_i \subseteq U_i$, we have that $\cap_i V_i \subseteq \cap_i U_i$. 
On the other hand, let $x_{V_i}: = \psi_\varepsilon(V_i)$ (i.e, $x_{V_i}$ is the representative of the set $V_i$). Note that $W_i := \pi_{\varepsilon}(x_{V_i})$ satisfies $W_i \supseteq B_\varepsilon(x_{V_i}) \supseteq V_i$.  
We then have that $\cap_i V_i \subseteq \cap_i U_i \bigcap \cap_j W_j$. Thus $\cap_i U_i \bigcap \cap_j W_j \neq \emptyset$, implying that $\{U_i\} \cup \{W_i\}$ spans a simplex in $\M_{c (s+2)\eps}$. Hence $s_{\eps,c(s+2)\eps}$ and $\pi_{\varepsilon} \circ \psi_{\varepsilon}$ are contiguous. 
\end{proof}

\begin{claim}\label{prop:commute-2}
The maps $\iota_{\eps,c(s+2)\eps}$ and $\psi_{\mycechconst \varepsilon}\circ \pi_{\varepsilon}$ in (\ref{eqn:cechMM-2}) are contiguous. 
\end{claim}
\begin{proof}

Specifically, consider $\sigma = \{ x_0, \ldots, x_k \} \in \mC_\varepsilon$, and let $V_i := \pi_{\varepsilon} (x_i)$ for each $i\in\{0,\ldots,k\}$. By definition of $\pi_\varepsilon$, we have that $B_\varepsilon(x_i) \subseteq V_i\,\,\mbox{for $i\in\{0,\ldots,k\}.$}$

Let $\tilde x_i := \psi_{\mycechconst \varepsilon}(V_i)$ for each $i\in\{0,\ldots,k\}$. Then, since $V_i$ belongs to $f^\ast(\mathcal{U}_{c(s+2)\eps})$, by definition of $\tilde x_i$ and Lemma \ref{lemma:ball-dUf}, we see that $ V_i \subset B_{\mycechconst \varepsilon}(\tilde x_i)$, for each $i \in \{0,\ldots, k\}$
, and thus 
$$B_\varepsilon(x_i)\subset B_{\mycechconst \varepsilon}(\tilde x_i)\,\,\mbox{for $i\in\{0,\ldots,k\}.$}$$

On the other hand one trivially has $B_\varepsilon(x_i)\subset B_{\mycechconst \varepsilon}(x_i)\,\,\mbox{for $i\in\{0,\ldots,k\}$},$
so  that then $B_\varepsilon(x_i)\subset B_{\mycechconst \varepsilon}(\tilde x_i)\cap  B_{\mycechconst \varepsilon}(x_i)$
for $i\in\{0,\ldots,k\}.$ Thus, 

$$\emptyset\neq \cap_i B_\varepsilon(x_i)\subseteq \bigcap_i B_{\mycechconst \varepsilon}(x_i)  \cap  B_{\mycechconst \varepsilon}(\tilde x_i). $$

Hence vertices in $\iota_{\varepsilon,\mycechconst \varepsilon}(\sigma) \cup \psi_{\mycechconst \varepsilon} \circ \pi_{\varepsilon} (\sigma)$ span a simplex in $\mC_{\mycechconst \varepsilon}$ establishing
that the two maps $\iota_{\eps,c(s+2)\eps}$ and $\psi_{\mycechconst \varepsilon} \circ \pi_{\varepsilon}$ are contiguous. 
\end{proof}

The following two claims are also true. Their proofs use arguments similar to the ones given in the proof of Claims \ref{prop:commute-1} and \ref{prop:commute-2} above and are ommitted.
\begin{claim}\label{prop:commute-3}
The maps $\pi_{\varepsilon'}\circ\iota_{\eps,\eps'}$ and $s_{\mycechconst \varepsilon,\mycechconst \varepsilon'}\circ \pi_{\varepsilon}$ in (\ref{eqn:cechMM-3}) are contiguous. 
\end{claim}
\begin{claim}\label{prop:commute-4}
The maps $\iota_{\eps,\eps'}\circ\psi_\eps$ and $\psi_{\eps'}\circ \iota_{\eps,\eps'}$ in (\ref{eqn:cechMM-4}) are contiguous. 
\end{claim}

The rest of the proof of Theorem \ref{theo:inter-metric} follows steps similar to those in the proof of Proposition \ref{prop:stab-func}. 
\end{proof}


\section{Discussion}
In this paper, we proposed, as well as studied theoretical and computational aspects of \emph{ multiscale mapper,} a construction which produces a multiscale summary of a map on a domain using a cover of the codomain at different scales. 

Given that in practice, hidden domains are often approximated by a set of discrete point samples, an important future direction will be to investigate how to approximate the multiscale mapper of a map $f: X \to Z$ on the metric 
space $X$ from a finite set of samples $P$ lying on or around $X$ as well as function values $f: P \to Z$ at points in $P$. It will be particularly interesting to be able to handle noise both in point samples $P$ and in the observed values of the function $f$.

We also note that the multiscale mapper framework can be potentially extended to a \emph{zigzag}-tower of covers, which will further increase the information encoded in the resulting summary. It will be interesting to study the theoretical properties of such a \emph{zigzag} version of mapper. Its stability however appears challenging, given that the theory of  stability of zigzag persistence modules is much less developed than that for standard persistence modules. 
Another interesting question is to understand the continuous object that the mapper converges to as the scale of the cover tends to zero. In particular, does the mapper converge to the Reeb space \cite{reeb-space}?  

Finally, it seems of interest to understand the features of a topological space which are captured by Multiscale Mapper and its variants. Some related work in this direction in the context of Reeb graphs has been reported 
recently \cite{justin,munch}.

\subsection*{Acknowledgments.} This work is partially supported by 
the National Science Foundation under grants 
CCF-1064416, CCF-1116258, CCF1526513, IIS-1422400, CCF-1319406, and CCF-1318595. 

\appendix

\section{The instability of Mapper}
\label{appendix:instability-mapper}
In this section we briefly discuss how one may perceive that the simplicial complexes produced by Mapper may not admit a simple notion of stability. 

As an example consider the situation in Figure \ref{fig:counter-mapper-1}.  Consider for each $\delta>0$ the domain $X_\delta$ shown in the figure (a topological graph with one loop), and the functions depicted in figure \ref{fig:counter-mapper-1}: these are height functions $f_\delta$ and $g_\delta$  which differ by $\delta$, that is $\|f_\delta-g_\delta\|_{\infty}=\delta$. The open cover is shown in the middle of the figure. Notice that the Mapper outputs for these two functions w.r.t. the same open cover of the co-domain $\mathbb{R}$ are different and that the situation can be replicated for each $\delta>0$.

In contrast, one of the features of the Multiscale Mapper construction is that it is amenable to a certain type of stability under changes in the function and in the tower of covers.  This situation is not surprising and is reminiscent of the pattern arising when comparing standard homology (and Betti numbers) computed on fixed simplicial complexes to persistent homology (and Betti barcodes/persistent diagrams) on tower of simplicial complexes (i.e. filtrations).

In what follows we will introduce a particular class of towers of covers that will be used to express some stability properties enjoyed by Multiscale Mapper.

\begin{figure}
\begin{center}
\includegraphics[width=\linewidth]{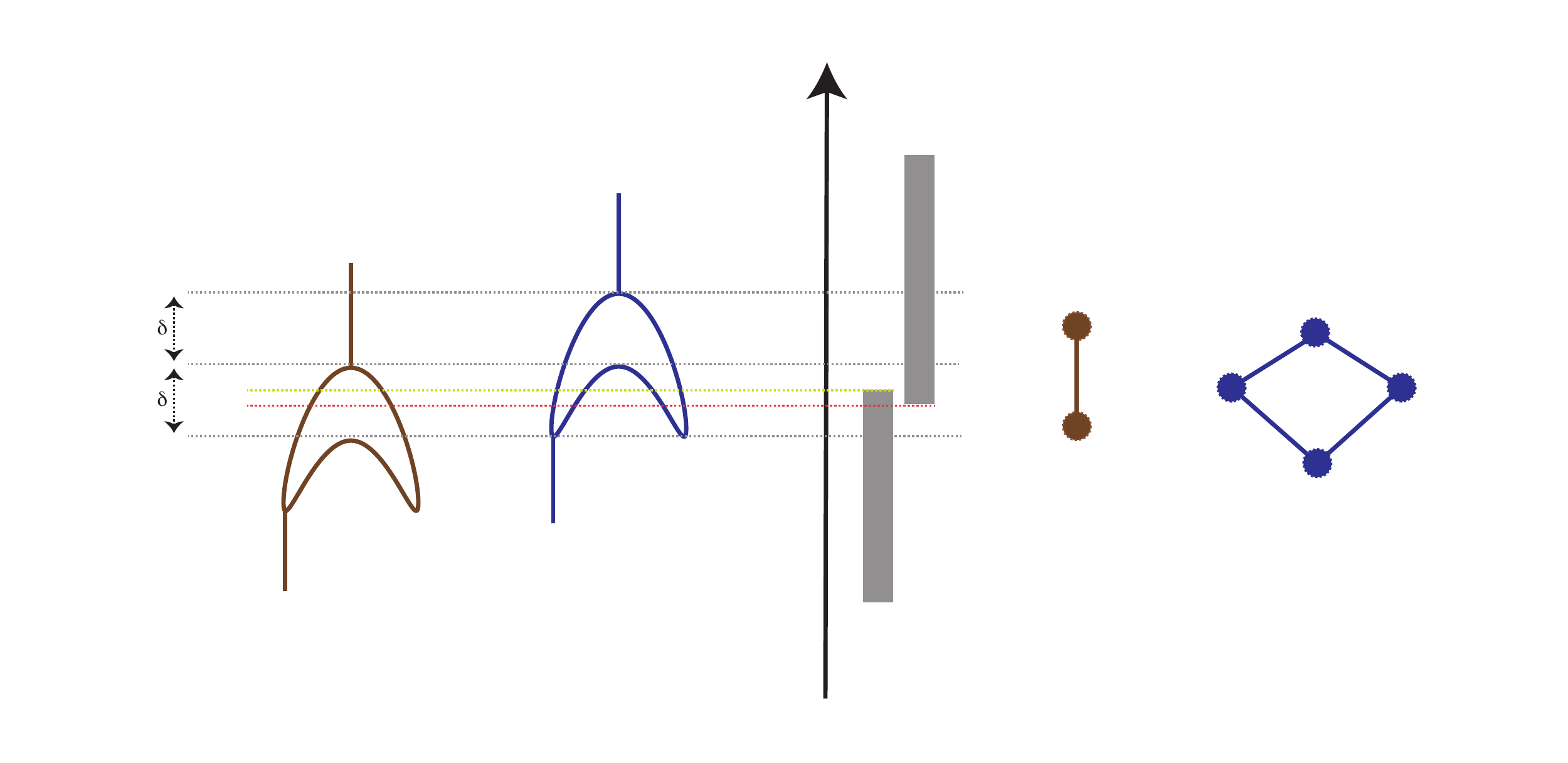}
\end{center}
\caption{A situation in which Mapper yields two different answers for similar functions on the same domain. The construction can be carried out for each $\delta>0$.}
\label{fig:counter-mapper-1}
\end{figure}

\section{Good towers of covers}
\label{appendix:subsec:goodcover}

\begin{remark}
We make the following remarks about the notion of $(c,s)$-good towers of covers:
\begin{itemize}
\item The underlying intuition is that an ``ideal'' tower of covers is one for which $c=1$ and $s=0$.
\item A first example is given by the following construction: pick any $s\in [0,\diam(Z)]$ and for $\varepsilon\geq s$ let $\mathcal{U}_\varepsilon:=\{B(z,\frac{\varepsilon}{2}),\,z\in Z\}$, i.e. the collection of \emph{all} $\varepsilon$-balls in $Z$. The maps $w_{\varepsilon,\varepsilon'}$ for $\varepsilon'\geq \varepsilon$ are defined in the obvious way: $B(z,\frac{\varepsilon}{2})\mapsto B(z,\frac{\varepsilon'}{2})$, all $z\in Z$. Then, since any $O\subset Z$ is contained in some ball $B(o,\diam(O))$ for some $o\in O$, this means that $\{\mathcal{U}_\varepsilon\}_{\varepsilon\geq s}$ is $(2,s)$-good. {Of course whenever $Z$ is path connected and not a singleton, this tower of covers is infinite.} A related finite construction is given next.
\item 
A discrete set $P\subset Z$ is called an \emph{$\initsampling$-sample of 
$(Z, d_Z)$} if for any point $z \in Z$, $d_Z(z, P) \le \initsampling$. 
Given a finite $\nu$-sampling $P$ of $Z$ one can always find a $(3,2\nu)$-good 
tower of covers of $Z$ consisting, for each scale $\varepsilon\geq 2\nu$, of all balls $B(p,\frac{\varepsilon}{2})$, $p\in P$. Details are given in Appendix \ref{appendix:sec:goodcovers}. This family may not be space efficient. An example of how to construct a space efficient $(c,s)$-good tower of covers for a compact metric space $(Z, d_Z)$ can be found in Appendix \ref{appendix:sec:goodcovers}. 
  \item Conditions (1) and (2) mean that the tower begins at resolution $s$, and the resolution parameter $\varepsilon$ controls the geometric characteristics of the elements of the cover. Parameter $s$ is clearly related to the size/finiteness of $\mathfrak{U}$: if $s$ is very small, then the number of sets in $\mathcal{U}_\varepsilon$ for $\varepsilon$ larger than but close to $s$ has to be large. 
Also, requiring $s$ to be smaller than or equal to the diameter of $Z$ ensures that condition 3 is not vacuously satisfied.

\item Condition (3) controls the degree to which one can inject a given set in $Z$ inside an element of the cover. The situation when $c > 1$ is consistent with having finite covers for all $\varepsilon$. On the other hand, $c=1$ may require, for each $\varepsilon$, open covers with infinitely many elements.

\item If $\mathfrak{U}$ is $(c,s)$-good, then it is $(c',s')$-good for all $c'\geq c$ and $s'\geq s$.
\end{itemize}
\end{remark}

\subsection{The necessity of $(c,s)$-good tower of covers.}
\label{subsec:badcover}

In Section \S\ref{subsec:stability}, we aim to obtain a stability result for the multiscale mapper $\mathrm{MM}(\mathfrak{U},f)$ with respect 
to perturbations of $f$. 
First, notice that we need to bound the resolution
of an tower of covers from below by a parameter $s>0$ to ensure its finiteness.  
We explain now why we need the second parameter $c$ for
the stability result.

Specifically, in what follows, we provide an example of a tower of covers of $\R$ which does not satisfy the condition
3 in $(c,s)$-goodness. This causes the persistence diagram $\mathrm{D}_1\mathrm{MM}(\mathfrak{U},f)$ to be \textbf{unstable} with respect to small perturbations of $f$.

\subsection*{Construction of a pathological tower of covers.} Given $M>0$ let $\mathbb{I}:=[-M,M].$ For each $\varepsilon>0$ and $k\in\Z$ let $I_k^{(\varepsilon)} := [k \cdot 2^{\lfloor \log_2^\eps \rfloor}, (k+1)\cdot 2^{\lfloor \log_2^\eps \rfloor} ]$. Pick any $s\geq 0$ and consider the following tower $\mathfrak{W}$  of covers of $\mathbb{I}$ by closed intervals:
$$
\mathfrak{W} = \big\{ \mathcal{W}_{\eps}\overset{\tiny{w_{\eps,\eps'}}}{\longrightarrow} \mathcal{W}_{\eps'}\big\}_{s \le \varepsilon\leq \varepsilon'} ;$$
where 
$$\mathcal{W}_\eps := \big\{ I_k^{(\varepsilon)}~\mid~ k\in \mathbb{Z}~\mbox{s.t. $I_k^{(\varepsilon)}\cap \mathbb{I}\neq \emptyset$} \big \}. 
$$

The maps of covers $w_{\eps, \eps'}: \mathcal{W}_\eps \to \mathcal{W}_{\eps'}$ are defined below.
\begin{claim} For any $\eps' >\eps$, and each element $I_k^{(\eps)}\in \mathcal{W}_\eps$, there exists a unique $k'\in \Z$ such that $I_{k}^{(\varepsilon)}\subseteq I^{(\varepsilon')}_{k'} $.
\end{claim}
For each $k\in\Z$ we set $w_{\eps, \eps'}(I_k^{(\varepsilon)}) = I_{k'}^{(\varepsilon')}$ where $k'\in\Z$ is given the claim above. 
\begin{proof}(Proof of the claim)
Write $q={\lfloor \log_2^{\eps} \rfloor}$ and $q'={\lfloor \log_2^{\eps'} \rfloor} = q+n$ for some $n\in\N$.
We are going to produce $k'\in\Z$ such that $k'\cdot 2^{q'}\leq k\cdot 2^q$ and $(k+1)\cdot 2^q\leq (k'+1)\cdot 2^{q'}$, which will immediately imply that $I_{k}^{(\varepsilon)}\subseteq I^{(\varepsilon')}_{k'} $.

Choose $k'$ to be the largest integer such that $k' \cdot 2^{q'}\le k \cdot 2^q$. This means that $k'\in\Z$ is maximal amongst integers for which $k'\cdot 2^n \leq k$. By maximality, we have that $(k'+1)\cdot 2^n \geq k+1$. Then, this means that $(k'+1) \cdot 2^{q'}\geq (k+1) \cdot 2^q$, which implies that indeed $I_{k}^{(\varepsilon)}\subseteq I^{(\varepsilon')}_{k'} $.

In fact, $I^{(\varepsilon')}_{k'} $ is the unique interval in $\mathcal{W}_{\eps'}$ that contains $I^{(\varepsilon)}_{k} $ due to the fact that intervals in $\mathcal{W}_{\eps'}$ have disjoint interiors.  
\end{proof}

\begin{remark}
By the argument in the proof of the above claim, it also follows that for any $s\leq \eps <\eps' < \eps''$ and any $I \in \mathcal{W}_\eps$, $w_{\eps, \eps''} (I) = w_{\eps', \eps''} (w_{\eps, \eps'} (I))$; that is, $w_{\eps, \eps''} = w_{\eps', \eps''} \circ w_{\eps, \eps'}$. Hence the maps $\{w_{\eps, \eps'}:\mathcal{W}_\varepsilon\rightarrow \mathcal{W}_{\varepsilon'}\}_{s\le \eps \le \eps'}$ are valid maps of covers, satisfying the conditions in Definition \ref{def:hfc}. Hence $\mathfrak{W}$ is a tower of (closed) covers. 
\end{remark}

\begin{remark}
This natural tower of covers $\mathfrak{W}$ is not $(c,s)$-good for any constant $c$. Specifically, consider an arbitrary small interval $(-r, r)$ for any $r > 0$. Clearly, there is no element in any $\mathcal{W}_\eps$ that contains this interval. 
It turns out that the persistence diagram arising from multiscale mapper is unstable w.r.t. perturbations of the input function, as we will show by an example shortly. Note that, in contrast, stability of the persistence diagrams of multiscale mapper outputs is guaranteed for $(c,s)$-good tower of covers, as stated in  Corollary \ref{coro:stab-func} and Theorem \ref{theo:stab-general}. 
\end{remark}

\subsection*{An example.}

\begin{figure}
\centerline{\includegraphics[height=6cm]{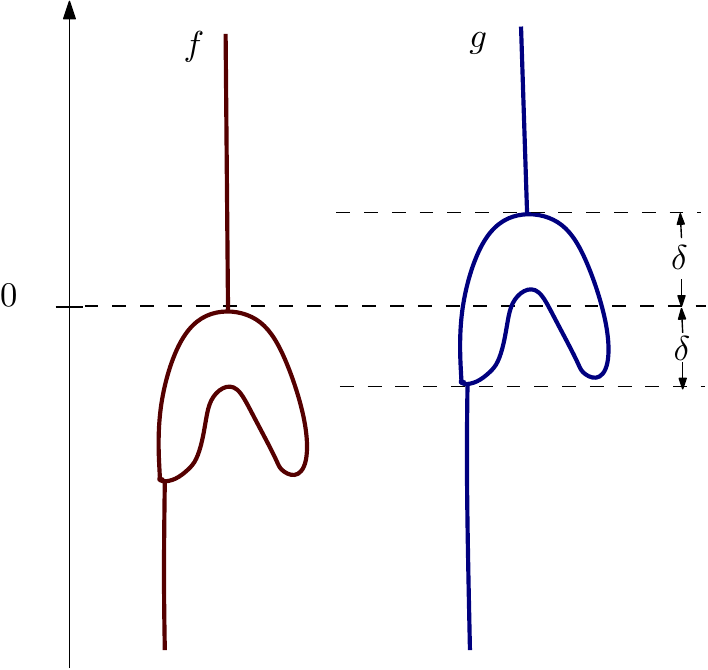}}
\caption{$f$ and $g$.}
\label{fig:twofuncs}
\end{figure}

We show how instability of multiscale mapper may arise for some choices of $\mathfrak{W}$ using the following example which is similar to the one presented in Figure \ref{fig:counter-mapper-1}.  Let $\delta$ be any value larger than $s$, the smallest scale of the tower of covers $\mathfrak{W}$.  Suppose we are given two functions defined on a graph $G$: $f, g: G \to \mathbb{R}$, as shown in Figure \ref{fig:twofuncs}. It is clear that $\| f -g \|_\infty = \delta$.

$\mathrm{D}_1\mathrm{MM}(\mathfrak{W},f)$ consists of the point $(s, 2\delta)$, indicating that a non-null homologous loop exists at scale $s$ in the pullback nerve complex $N(f^*(\mathcal{W}_s))$, but is killed at scale $2\delta$ in the nerve complex $N(f^*(\mathcal{W}_{2\delta}))$ (specifically, the image of this loop under the simplicial map $f^*(w_{s,2\delta})$ becomes null-homologous).

However, for the function $g$, there is a loop created at the lowest scale $s$, but the homology class carried by this loop is never killed under the simplicial maps $g^*(w_{s, \eps'}): N(g^*(\mathcal{W}_s)) \to N(g^*(\mathcal{W}_{\eps'})$ for any $\eps' > s$. Thus  the persistence diagram $\mathrm{D}_1\mathrm{MM}(\mathfrak{W},g)$ consists of the point $(s, \infty)$. 
Hence the two persistence diagrams $\mathrm{D}_1\mathrm{MM}(\mathfrak{W},f)$ and $\mathrm{D}_1\mathrm{MM}(\mathfrak{W},g)$ are not close under the bottleneck distance (in fact, the bottleneck distance between them is $\infty$),  despite the fact that the functions $f$ and $g$ are $\delta$-close , that is $\| f - g \|_\infty \le \delta$.

\begin{remark}
We remark that for clarity of presentation, in the above construction, each element in the covering $\mathcal{W}_\eps$ is a closed interval. However, this example can be easily extended to open covers: Specifically, let $\nu$ be a sufficiently small positive value such that $\nu < s < \delta$. Then we can change each interval $I = [k \cdot 2^{\lfloor \log_2^\eps \rfloor}, (k+1)\cdot 2^{\lfloor \log_2^\eps \rfloor} ] \in \mathcal{W}_\eps$ to $I^\nu = (k \cdot 2^{\lfloor \log_2^\eps \rfloor}-\nu, (k+1)\cdot 2^{\lfloor \log_2^\eps \rfloor} +\nu)$. Note that for an arbitrary small interval $(-r, r)$ with $r > \nu$, there is no element in any $\mathcal{W}_\eps$ that contains it, so that the resulting tower of covers is not $(c,s)$-good for any $c\geq 1$. Hence the example described in Figure \ref{fig:twofuncs} can be adapted whenever $\delta > \nu$; this leads to the following statement:
\end{remark}

\begin{proposition}
Fix $s\geq 0$ and $W>s$. Then, for each $\frac{W}{2} > \delta>0$ there exist (1) a topological graph $G_\delta$, (2) two continuous functions $f_\delta,g_\delta:G_\delta\rightarrow [-W,W]$, and (3) a tower of covers $\mathfrak{W}$ of $[-W,W]$ satisfying the following properties:
\begin{itemize}
\item $\|f_\delta-g_\delta\|_{\infty}=\delta$,
\item $\mathrm{D}_1\MM(\mathfrak{W},f_\delta)=\{(s,2\delta)\}$ and $\mathrm{D}_1\MM(\mathfrak{W},g_\delta)=\{(s,\infty)\}.$
\end{itemize}
\end{proposition}

\subsection{Constructing a good towers of covers.}
\label{appendix:sec:goodcovers}

\newcommand{\basedelta}		{\rho}
\newcommand{\Net}				{{\mathcal{N}}}
\newcommand{\aball}				{{\mathrm{B}}}

Suppose we are given a compact metric space $(Z,d_Z)$ with bounded doubling dimension. We assume that we can obtain an 
\emph{$\initsampling$-sample $P$ of $(Z, d_Z)$}, which is a discrete set of points $P \subset Z$ such that for any point $z \in Z$, $d_Z(z, P) \le \initsampling$. 
For example, if the input metric space $Z$ is a $d$-dimensional cube in the Euclidean space $\mathbb{R}^d$, we can simply choose $P$ to be the set of vertices from a $d$-dimensional lattice with edge length $\initsampling$.  For simplicity, assume that $\initsampling \le 1$ (otherwise, we can rescale the metric to make this hold). 

\subsubsection{A simple ($c,s$)-good tower of covers.} \label{simple-covs}

Consider the following tower of covers $\mathfrak{W} = \{ \mathcal{W}_\eps \mid \eps \ge 2 \initsampling \}$ where 
$\mathcal{W}_\eps := \{ \aball_{\frac{\eps}{2}}(u) \mid u \in P \}$. The associated maps of covers $w_{\eps, \eps'}: \mathcal{W}_\eps \to \mathcal{W}_{\eps'}$ simply sends each element $\aball_{\frac{\eps}{2}}(u)\in \mathcal{W}_\eps$ to the corresponding set $\aball_{\frac{\eps'}{2}}(u) \in \mathcal{W}_{\eps'}$. 
It is easy to see that $\mathfrak{W}$ is $(3, 2\initsampling)$-good. 
Indeed, given any $O \subseteq Z$ with diameter $R = \diam(O) \ge s = 2\initsampling$, pick an arbitrary point $o \in O$ and let $u \in P$ be a nearest neighbor of $o$ in $P$. 
We then have that for any point $x\in O$, 
$$
d_Z(x, u) \le d_Z(x,o) + d_Z(o, u) \le R + \initsampling \le \frac{3R}{2}. $$
That is, $O \subset \aball_{\frac{3R}{2}}(u) \in \mathcal{W}_{3R}$. 

This tower of covers however has large size. In particular, as the scale $\eps$ becomes large, the number of elements in $\mathcal{W}_\eps$ remains the same, while intuitively, a much smaller subset will be sufficient to cover $Z$. 
In what follows, we describe a different construction of a good tower of covers based on using the so-called \emph{nets}. 

\subsubsection{A space-efficient ($c,s$)-good tower of covers.} \label{space-eff-covers}
Following the notations used by Har-Peled and Mendel in \cite{HM06}, we have the following: 

\begin{proposition}[\cite{HM06}]\label{def:nets}
For any constant $\basedelta \ge 11$, and for any scale $\ell \in \mathbb{R}^+$, one can compute a $\rho^\ell$-net $\Net(\ell) \subseteq P$ in the sense that 
\begin{itemize}\denselist
\item [(i)] for any $p \in P$, $d_Z(p, \Net(\ell)) \le \basedelta^\ell$; 
\item [(ii)] any two points $u, v \in \Net(\ell)$, $d_Z(u, v) \ge \basedelta^{\ell-1}/16$, and 
\item [(iii)] $\Net(\ell') \subseteq \Net(\ell)$ for any $\ell < \ell'$; that is, the net at a bigger scale is a subset of net at a smaller scale. 
\end{itemize}
Each $\Net(\ell)$ is referred to as a net at scale $\ell$. 
\end{proposition}

From now on, we set $\eps_i := 4(\rho+1)^i$ for any positive integer $i \in \mathbb{Z}^+$. 
Consider the collection of nets $\{ \Net(i) \mid i \in \mathbb{Z}^+\}$ where $\Net(i)$ is as described in Proposition \ref{def:nets}.

\begin{definition}\label{def:HC-nets}
We define a tower of covers $\mathfrak{U} = \{U_{\eps_i} \mid i \in \mathbb{Z}^+ \}$ where: 
\begin{equation}\label{eqn:Ui}
U_{\eps_i} := \big\{ \aball_{\frac{\eps_i}{2}}(u) \mid u \in \Net(i) \}, ~\text{where}~\eps_i = 4(\rho+1)^i \big\}.
\end{equation}

The associated maps of covers $u_{\eps_i, \eps_j}: \mathcal{U}_{\eps_i} \to \mathcal{U}_{\eps_j}$, for any $i ,j \in \mathbb{Z}^+$ with $i\neq j$, are defined as follows: 
\begin{itemize}\denselist
\item[(1)] $u_{\eps_i, \eps_{i+1}}: \mathcal{U}_{\eps_i} \to \mathcal{U}_{\eps_{i+1}}$: For any element $U = \aball_{\frac{\eps_i}{2}}(u) \in \mathcal{U}_{\eps_i}$, we find the nearest neighbor $v \in \Net(i+1)$ of $u$ in $\Net(i)$, and map $U$ to $V = \aball_{\frac{\eps_{i+1}}{2}}(v) \in \mathcal{U}_{\eps_{i+1}}$. 
\item[(2)] For $i < j-1$, we set $u_{\eps_i, \eps_j}$ as the concatenation $u_{\eps_{j-1}, \eps_j} \circ u_{\eps_{j-2},\eps_{j-1}} \circ \cdots \circ u_{\eps_{i}, \eps_{i+1}}$. 
\end{itemize}
\end{definition}

\begin{remark}
We note that the above tower of covers $\mathfrak{U}$ is discrete in the sense that we only consider covers $\mathcal{U}_\eps$ for a discrete set of $\eps$s from $\{ \eps_i \mid i\in \mathbb{Z}^+\}$. However, one can easily extend it to a tower of covers $\mathfrak{U}^\mathrm{ext}=\{\mathcal{U}_\delta^\mathrm{ext}\}$ which is defined for all $\delta\in\R_+$ by declaring that $\mathcal{U}_\delta^\mathrm{ext} = \mathcal{U}_{\eps_{\zeta(\delta)}}$ where $\zeta(\delta) = \max\{i\in\Z|\,\varepsilon_i\leq \delta\}$. One may define the cover maps in $\mathfrak{U}^\mathrm{ext}$ in a similar manner. For simplicity of exposition, in what follows we use a discrete tower of covers. 
\end{remark}

\begin{claim}\label{claim:coverZ}
For any $i\in \mathbb{Z}^+$, $U_{\eps_i}$ forms a covering of the metric space $Z$ (where $P$ are sampled from). 
\end{claim}
\begin{proof}
$P$ is an $\initsampling$-sampling of $(Z, d_Z)$ with $\initsampling \le 1$. Hence for any point $x\in Z$, there is a point $p_x \in P$ such that $d(x, p_x) \le \eps$. 
By Property (i) of Definition \ref{def:nets}, $d_Z(p_x, \Net(i)) \le \basedelta^i$. 
Hence $d_Z(x, \Net(i)) \le d_Z(x,p_x) + d_Z(p_x, \Net(i)) \le 1 + \basedelta^i \le (\basedelta+1)^i$. As such, there exists $u \in \Net(i)$ such that $d_Z(x, u) \le (\basedelta+1)^i$; that is, $x\in U =\aball_{\frac{\eps_i}{2}}(u) \in \mathcal{U}_{\eps_i}$. 
\end{proof}

\begin{claim}\label{claim:inclusion}
For any $i < j$ and any $U \in \mathcal{U}_{\eps_i}$, we have $U \subseteq u_{\eps_i, \eps_j}(U)$. 
\end{claim}
\begin{proof}
We show that the claim holds true for $j = i+1$, and the claim then follows from this and the construction of $u_{\eps_i, \eps_j}$ for $i < j-1$. 

Suppose $U = \aball_{\frac{\eps_i}{2}}(u)$ for $u \in \Net(i)$, and let $V = u_{\eps_i, \eps_{i+1}}(U)$. 
By construction, $V = \aball_{\frac{\eps_{i+1}}{2}}(v)$ is such that $v$ is the nearest neighbor of $u$ in $\Net(i+1)$. 
If $u = v$, then clearly $U \subseteq V$. 

If $u \neq v$, by Property (i) of Definition \ref{def:nets}, $d_Z(u, v) \le \basedelta^{i+1}$. 
At the same time, any point $x\in U$ is within $\frac{\eps_i}{2} = 2 (\basedelta+1)^i$ distance to $u$. 
Hence 
$$d_Z(x, v) \le d_Z(x, u) + d_Z(u, v) \le 2(\basedelta+1)^i + \basedelta^{i+1} < 2(\basedelta+1)^{i+1}. $$
It then follows that $U \subseteq V$. 
\end{proof}

\begin{theorem}\label{thm:netsgivegoodcover}
$\mathfrak{U}$ as constructed above is a $(c, s)$-good tower of covers with $c = s = 4(\basedelta+1)$.  
\end{theorem}
\begin{proof}
First, note that by Claim \ref{claim:coverZ}, each $\mathcal{U}_\eps \in \mathfrak{U}$ is indeed a cover for $(Z,d_Z)$. By Claim \ref{claim:inclusion}, the associated maps as constructed in Definition \ref{def:HC-nets} are valid maps of covers. Hence $\mathfrak{U}$ is indeed a tower of covers for $(Z, d_Z)$. 
We now show that $\mathfrak{U}$ is $(c,s)$-good. 

First, by Eqn. \ref{eqn:Ui}, each set $U$ in $\mathcal{U}_{\eps_i}$ obviously has diameter at most $\eps_i$. Also, $\mathfrak{U}$ is $s$-truncated for $s = 4(\basedelta+1)$ since it starts with $\eps_1 = s$. Hence properties 1 and 2 of Definition \ref{def:goodcover-II} hold. 
What remains is to show that property 3 of Definition \ref{def:goodcover-II} also holds. 

Specifically, consider any $O \subset Z$ such that $\diam(O) \ge s$. Set $R = \diam(O)$ and let $a$ be the unique integer such that $$(\basedelta+1)^{a-1} \le R < (\basedelta+1)^{a}. $$
Let $o\in O$ be any point in $O$, and $p\in P$ the nearest neighbor of $o$ in $P$: By the sampling condition of $P$, we have $d_Z(o,p) \le \nu < 1$. Let $u\in \Net(a)$ be the nearest neighbor of $p$ in the net $\Net(a)$.  By Definition \ref{def:nets} (i), $d_Z(p, u) \le \basedelta^a$. 
We thus have that, for any point $x \in O$, 
\begin{align*}
d_Z(x,u) &\le d_Z(x,o) + d_Z(o, p) + d_Z(p,u)\\
 &\le R + \eps + \basedelta^a\\ &< (\basedelta+1)^a + 1 + \basedelta^a\\ &\le 2(\basedelta+1)^a.
\end{align*}
In other words, $O \subseteq \aball_{2(\basedelta+1)^a}(u) = \aball_{\frac{\eps_a}{2}}(u) \in \mathcal{U}_{\eps_a}$. 

On the other  hand, since $R \ge (\basedelta+1)^{a-1}$, we have that for $c = 4(\basedelta+1)$, 
$$\eps_a = 4 (\basedelta+1)^a = c \cdot (\basedelta+1)^{a-1} \le c R = c \cdot \diam(O). $$ 

\end{proof}

\subsection*{Space-efficiency of $\mathfrak{U}$.}
In comparison with the simple $(3, 2\basedelta$)-good tower of covers $\mathfrak{W}$ that we introduced at the beginning of this section, the main advantage of $\mathfrak{U}$ is its much more compact size. Intuitively, this comes from property (ii) of Proposition \ref{def:HC-nets}, which states that the points in the net $\Net(i)$ are sparse and contain little redundancy. 
In fact, consider the covering $\mathcal{U}_{\eps_i}$ at scale $\eps_i$. It size (i.e, the cardinality of $\mathcal{U}_{\eps_i}$) is close to optimal in the following sense: 

For any $\eps$, denote by $\mathcal{V}^*(\eps)$ the smallest possible (in terms of cardinality) covering of $(Z, d_Z)$ such that each element $V \in \mathcal{V}^*(\eps)$ has diameter at most $\eps$. 
Now, let $s^*(\eps) := |\mathcal{V}^*(\eps)|$ denote this optimal size for any covering of $(Z, d_Z)$ by elements with diameter at most $\eps$.  

\begin{proposition} \label{prop:optsize}
For any $i \in \mathbb{Z}^+$, $s^*(\eps_i) \le |\mathcal{U}_{\eps_i}| = | \Net(i) | \le s^*(\frac{\eps_i}{16\basedelta})$. 
\end{proposition}
\begin{proof}
The left inequality follows from the definition of $s^*(\eps_i)$. 
We now prove the right inequality. 
Specifically, consider the smallest covering $\mathcal{V}^* = \mathcal{V}^*(\frac{\eps_i}{16\basedelta})$ with $s^*(\frac{\eps_i}{16\basedelta}) = |\mathcal{V}^*|$. 
By property (ii) of Proposition \ref{def:HC-nets}, each set $V \in \mathcal{V}^*$ can contain at most one point from $\Net(i)$ since the diameter of $V$ is at most $\frac{\eps_i}{16\basedelta}$. At the same time, we know that the union of all sets in $\mathcal{V}^*$ will cover all points in $\Net(i)$. The right inequality then follows. 
\end{proof}


\end{document}